\documentclass[a4paper,twocolumn,10pt,accepted=2023-06-16]{quantumarticle}
\pdfoutput=1
\usepackage[utf8]{inputenc}
\usepackage[english]{babel}
\usepackage[T1]{fontenc}

\usepackage{amsfonts, amssymb, amsmath, amsthm}
\usepackage{latexsym}
\usepackage[tracking=smallcaps]{microtype}	
\usepackage{url}
\usepackage{color}
\definecolor{DarkGray}{rgb}{0.1,0.1,0.5}
\usepackage[colorlinks=true,breaklinks, linkcolor=DarkGray,citecolor=DarkGray,urlcolor=DarkGray]{hyperref}	

\usepackage{graphicx}
\usepackage{subfigure}

\newcommand{\ket}[1]{{|#1\rangle}}

\DeclareMathOperator{\pa}{\operatorname{parent}}
\DeclareMathOperator{\ro}{\operatorname{root}}
\DeclareMathOperator{\spam}{\operatorname{SPAM}}

\newcommand{\identity}{\ensuremath{I}} 
\newcommand{\Id}{\identity} 
\DeclareMathOperator{\CNOT}{\operatorname{CNOT}}

\DeclareMathOperator{\depth}{\operatorname{depth}}

\newcounter{sprows}

\newlength{\spheight}
\newlength{\spraise}

\newlength{\commentslength}

\newcommand{\rem}[1]{}

\newtheorem{theorem}{Theorem}
\newtheorem{lemma}[theorem]{Lemma}

\newfont{\subsubsecfnt}{ptmri8t at 11pt}
\renewcommand{\subparagraph}[1]{\smallskip{\subsubsecfnt #1.}}

\newcommand{\eqnref}[1]{\hyperref[#1]{{(\ref*{#1})}}}
\newcommand{\thmref}[1]{\hyperref[#1]{{Theorem~\ref*{#1}}}}
\newcommand{\lemref}[1]{\hyperref[#1]{{Lemma~\ref*{#1}}}}
\newcommand{\corref}[1]{\hyperref[#1]{{Corollary~\ref*{#1}}}}
\newcommand{\defref}[1]{\hyperref[#1]{{Definition~\ref*{#1}}}}
\newcommand{\secref}[1]{\hyperref[#1]{{Section~\ref*{#1}}}}
\newcommand{\fullfigref}[1]{\hyperref[#1]{{Figure~\ref*{#1}}}}
\newcommand{\figref}[1]{\hyperref[#1]{{Fig.~\ref*{#1}}}}  
\newcommand{\tabref}[1]{\hyperref[#1]{{Table~\ref*{#1}}}}
\newcommand{\remref}[1]{\hyperref[#1]{{Remark~\ref*{#1}}}}
\newcommand{\appref}[1]{\hyperref[#1]{{Appendix~\ref*{#1}}}}
\newcommand{\claimref}[1]{\hyperref[#1]{{Claim~\ref*{#1}}}}
\newcommand{\factref}[1]{\hyperref[#1]{{Fact~\ref*{#1}}}}
\newcommand{\propref}[1]{\hyperref[#1]{{Proposition~\ref*{#1}}}}
\newcommand{\exampleref}[1]{\hyperref[#1]{{Example~\ref*{#1}}}}
\newcommand{\conjref}[1]{\hyperref[#1]{{Conjecture~\ref*{#1}}}}

\allowdisplaybreaks[1]

\def\COLOR{}
\ifdefined\COLOR

\else

\fi

\definecolor{Cayenne}{rgb}{0.5,0,0}
\definecolor{Midnight}{rgb}{0,0,0.5}
\definecolor{Plum}{rgb}{0.5,0,0.5}
\definecolor{Teal}{rgb}{0,0.5,0.5}
\definecolor{Clover}{rgb}{0,0.5,0}
\definecolor{Maroon}{rgb}{0.5,0,0.25}
\definecolor{Ocean}{rgb}{0,0.25,0.5}
\definecolor{Tangerine}{rgb}{1,0.5,0}
\definecolor{Strawberry}{rgb}{1,0,0.5}
\definecolor{Fern}{rgb}{0.25,0.5,0}
\definecolor{Aqua}{rgb}{0,0.5,1}
\definecolor{Moss}{rgb}{0,0.5,0.25}
\definecolor{Mocha}{rgb}{0.5,0.25,0}
\definecolor{Lemon}{rgb}{1,1,0}
\definecolor{Asparagus}{rgb}{0.5,0.5,0}
\definecolor{Grape}{rgb}{0.5,0,1}
\definecolor{Iron}{rgb}{.3,.3,.3}
\definecolor{Steel}{rgb}{.4,.4,.4}
\definecolor{Purple}{rgb}{.5,0,.5}
\definecolor{purple(x11)}{rgb}{0.63, 0.36, 0.94}

\usepackage{enumitem} 

\usepackage{listings}
\usepackage{multirow, tabularx}
\usepackage{lipsum}
\makeatletter
\let\save@mathaccent\mathaccent
\newcommand*\if@single[3]{%
  \setbox0\hbox{${\mathaccent"0362{#1}}^H$}%
  \setbox2\hbox{${\mathaccent"0362{\kern0pt#1}}^H$}%
  \ifdim\ht0=\ht2 #3\else #2\fi
  }
\newcommand*\rel@kern[1]{\kern#1\dimexpr\macc@kerna}
\newcommand*\widebar[1]{\@ifnextchar^{{\wide@bar{#1}{0}}}{\wide@bar{#1}{1}}}
\newcommand*\wide@bar[2]{\if@single{#1}{\wide@bar@{#1}{#2}{1}}{\wide@bar@{#1}{#2}{2}}}
\newcommand*\wide@bar@[3]{%
  \begingroup
  \def\mathaccent##1##2{%
    \let\mathaccent\save@mathaccent
    \if#32 \let\macc@nucleus\first@char \fi
    \setbox\z@\hbox{$\macc@style{\macc@nucleus}_{}$}%
    \setbox\tw@\hbox{$\macc@style{\macc@nucleus}{}_{}$}%
    \dimen@\wd\tw@
    \advance\dimen@-\wd\z@
    \divide\dimen@ 3
    \@tempdima\wd\tw@
    \advance\@tempdima-\scriptspace
    \divide\@tempdima 10
    \advance\dimen@-\@tempdima
    \ifdim\dimen@>\z@ \dimen@0pt\fi
    \rel@kern{0.6}\kern-\dimen@
    \if#31
      \overline{\rel@kern{-0.6}\kern\dimen@\macc@nucleus\rel@kern{0.4}\kern\dimen@}%
      \advance\dimen@0.4\dimexpr\macc@kerna
      \let\final@kern#2%
      \ifdim\dimen@<\z@ \let\final@kern1\fi
      \if\final@kern1 \kern-\dimen@\fi
    \else
      \overline{\rel@kern{-0.6}\kern\dimen@#1}%
    \fi
  }%
  \macc@depth\@ne
  \let\math@bgroup\@empty \let\math@egroup\macc@set@skewchar
  \mathsurround\z@ \frozen@everymath{\mathgroup\macc@group\relax}%
  \macc@set@skewchar\relax
  \let\mathaccentV\macc@nested@a
  \if#31
    \macc@nested@a\relax111{#1}%
  \else
    \def\gobble@till@marker##1\endmarker{}%
    \futurelet\first@char\gobble@till@marker#1\endmarker
    \ifcat\noexpand\first@char A\else
      \def\first@char{}%
    \fi
    \macc@nested@a\relax111{\first@char}%
  \fi
  \endgroup
}
\makeatother

\makeatletter
\newcommand\footnoteref[1]{\protected@xdef\@thefnmark{\ref{#1}}\@footnotemark}
\makeatother

\begin{document}

\title{Fault-tolerant syndrome extraction and cat state preparation with fewer qubits}

\author{Prithviraj Prabhu}
\email{prithvirajprab@gmail.com}
\affiliation{University of Southern California, Los Angeles, CA 90089, USA}
\orcid{0000-0002-2445-2701}
\author{Ben Reichardt}
\orcid{0000-0003-0290-4698}

\maketitle

\begin{abstract}
We reduce the extra qubits needed for two fault-tolerant quantum computing protocols: error correction, specifically syndrome bit measurement, and cat state preparation.  For distance-three fault-tolerant syndrome extraction, we show an exponential reduction in qubit overhead over the previous best protocol.  For a weight-$w$ stabilizer, we demonstrate that stabilizer measurement tolerating one fault needs at most $\lceil \log_2 w \rceil + 1$ ancilla qubits.  If qubits reset quickly, four ancillas suffice.  We also study the preparation of entangled cat states, and prove that the overhead for distance-three fault tolerance is logarithmic in the cat state size.  These results apply both to near-term experiments with a few qubits, and to the general study of the asymptotic resource requirements of syndrome measurement and state preparation.

With $a$ flag qubits, previous methods use $O(a)$ flag patterns to identify faults. In order to use the same flag qubits more efficiently, we show how to use nearly all $2^a$ possible flag patterns, by constructing maximal-length paths through the $a$-dimensional hypercube.
\end{abstract}

\section{Introduction}

A critical component of quantum error correction is syndrome measurement: a set of circuits that are used to pinpoint which qubits have errors.  This process of error identification is itself susceptible to noise and may fail.  To make this robust, extra (ancilla) qubits can be used to identify damaging mid-circuit faults and mitigate the spread of errors.  The objective of this paper is to reduce the overhead of ancilla qubits used in imparting this fault tolerance.  In particular, we focus on optimizing the flag technique for distance-three fault-tolerant stabilizer measurement.  We also reduce qubit overhead in distance-three fault-tolerant cat state preparation.  Cat states~\cite{Greenberger1989} have applications in many areas of quantum computing, including communication~\cite{Gisin07}, information processing~\cite{PanChen2012}, and error correction~\cite{Shor96}.  Besides practical applications, our results on cat state preparation are theoretically interesting since: $i$) we introduce the study of asymptotic estimates of qubit overhead for the fault-tolerant preparation of cat states of \textit{arbitrary} size, and, $ii$) ideas developed for cat state preparation may provide clues for the fault-tolerant preparation of logical states of more complex codes.

 \begin{figure}
 \centering
\subfigure[\label{f:genove}]{\includegraphics[width = 0.42\textwidth]{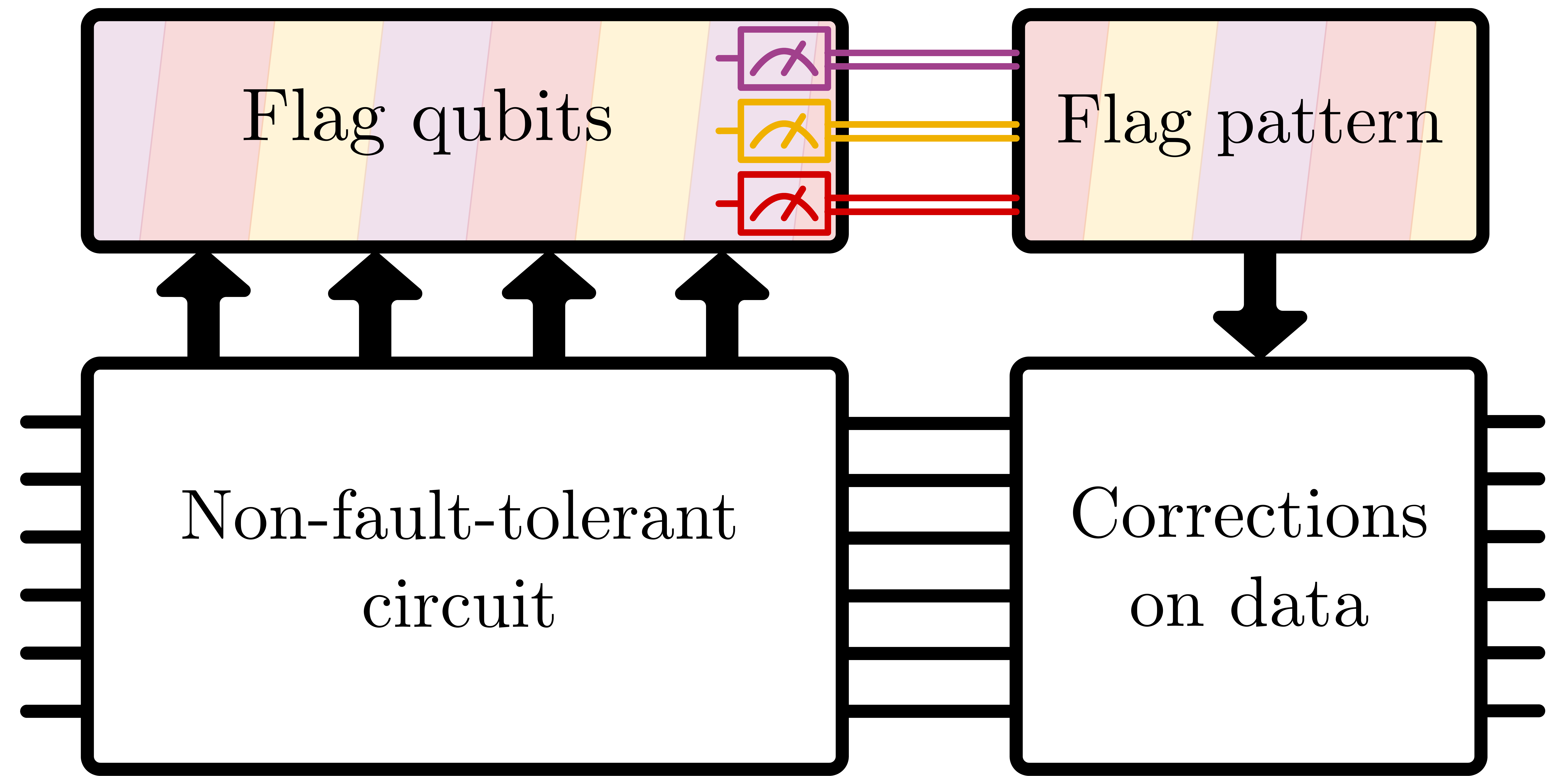}}
\subfigure[\label{f:syndromemeasurementd3slowresetw10}]{\includegraphics[width = 0.46\textwidth]{syndromemeasurementd3w10slowresetalt2}}
\caption{
(a)~Function of a flag scheme. Errors in a non-fault-tolerant circuit can be made to spread into flag qubits.  On measurement, the flag qubits yield a pattern of $1$s and~$0$s, based on which the data is corrected.  
(b)~Circuit to measure the stabilizer $X^{\otimes 10}$, using three flag qubits, in color, to protect against one $X$ fault (distance three).
}
\end{figure}

\begin{figure*}
\centering
\subfigure[\label{f:Shorcard} \cite{Shor96}]{\includegraphics[width = 0.24\textwidth]{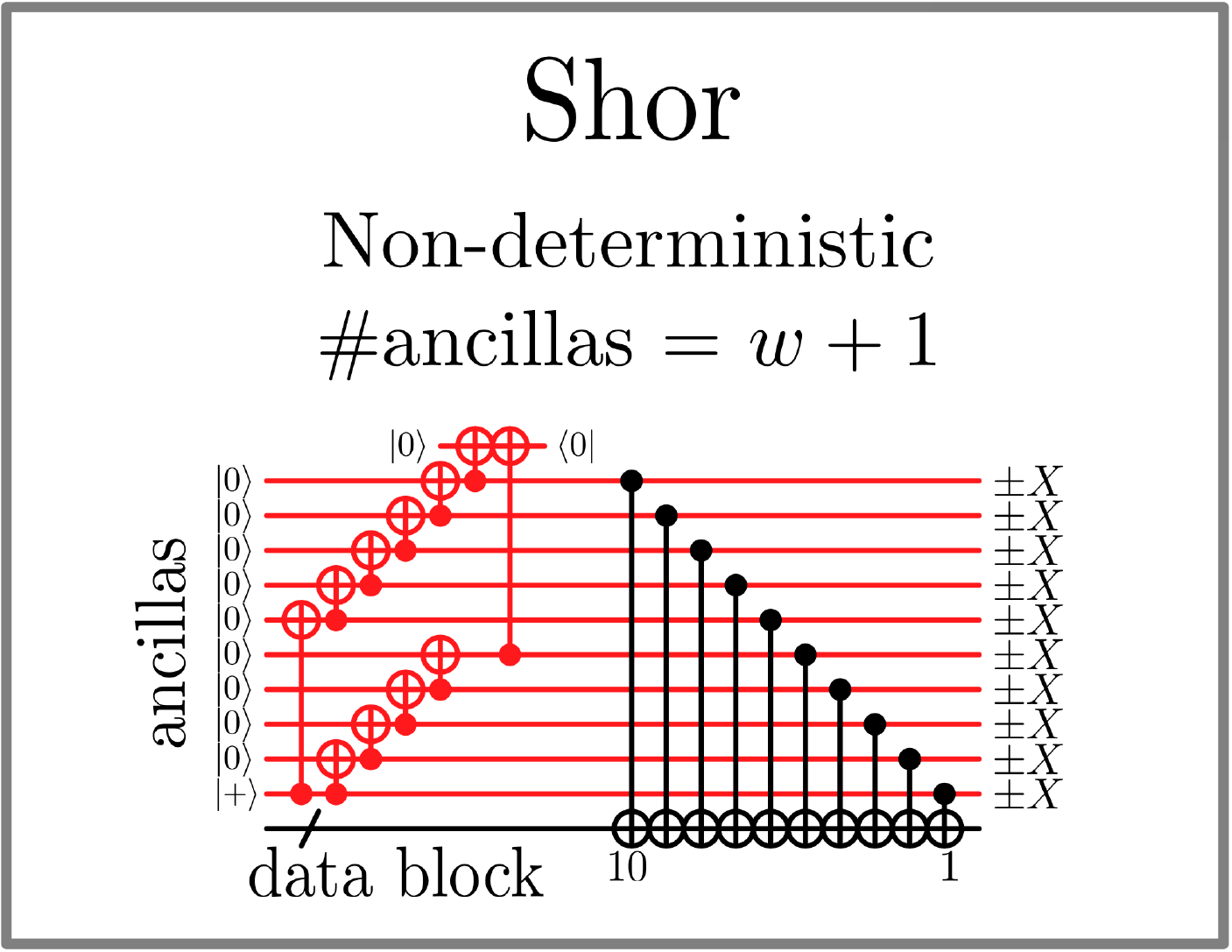}}
\subfigure[\label{f:DAcard} \cite{DiVincenzoAliferis06slow}]{\includegraphics[width = 0.24\textwidth]{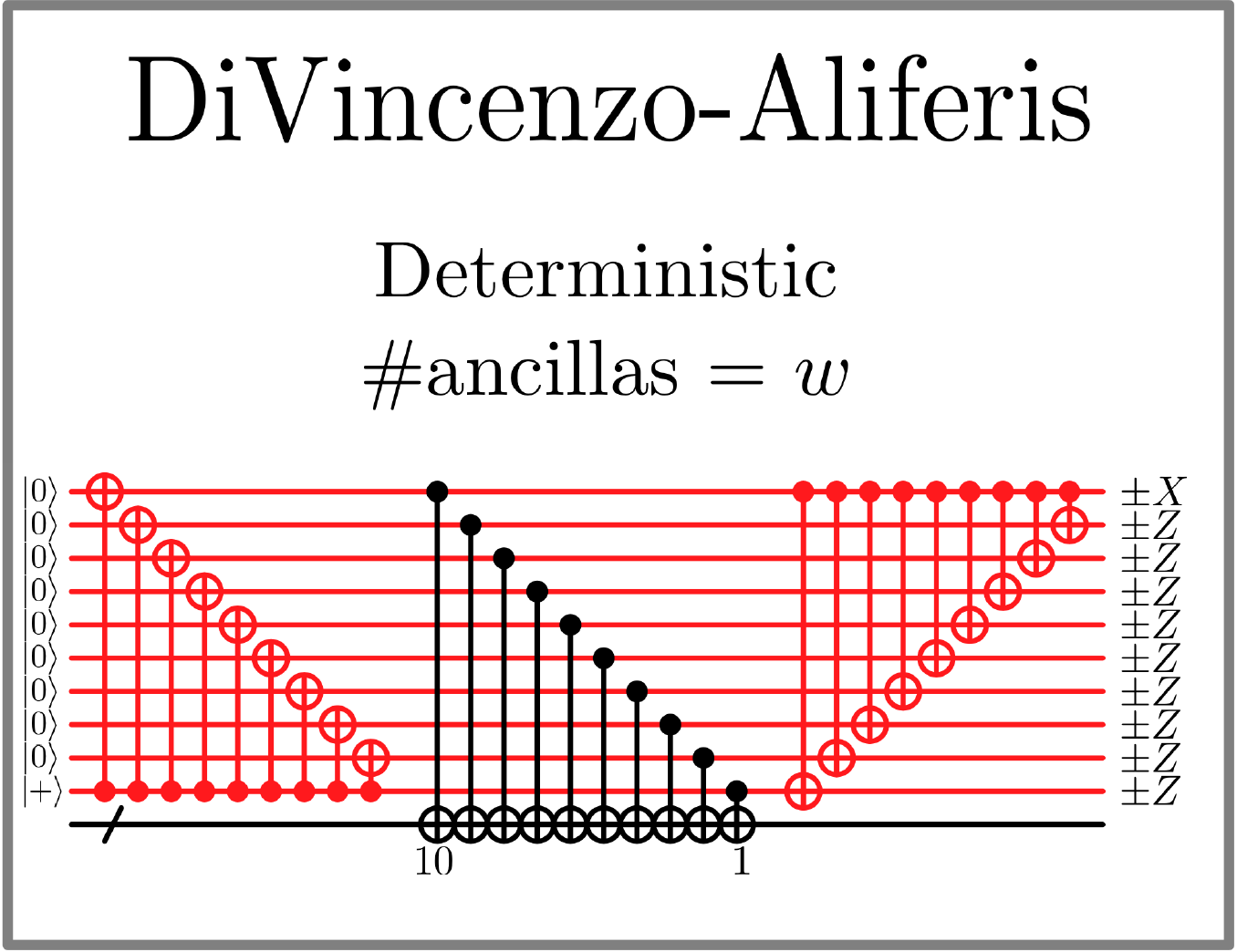}}
\subfigure[\label{f:CompDAcard} \cite{Stephens14colorcodeft, YoderKim16trianglecodes, ChaoReichardt18fewqubitcomputation}]{\includegraphics[width = 0.24\textwidth]{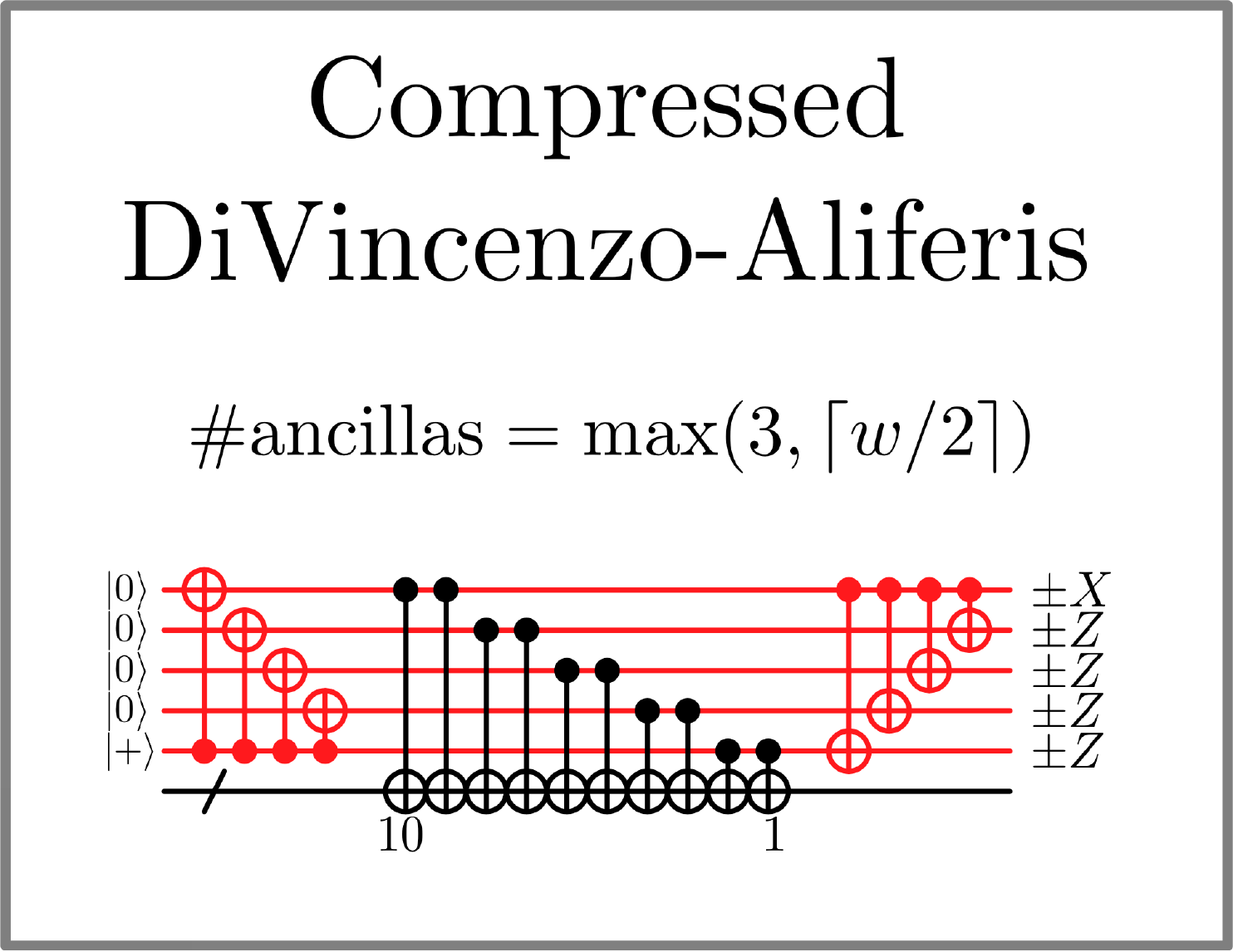}}
\subfigure[\label{f:Flagcard}]{\includegraphics[width = 0.24\textwidth]{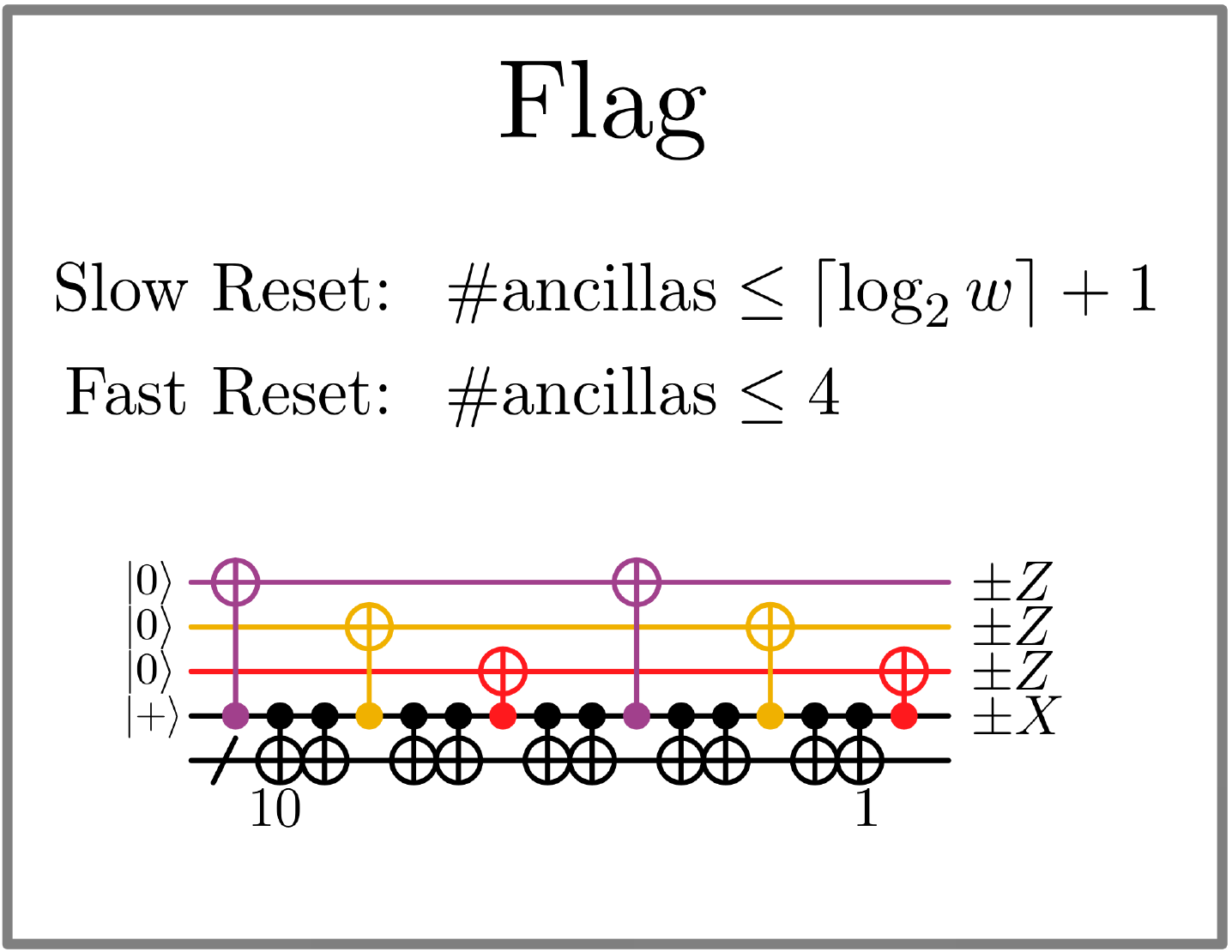}}
\caption{Historical progression of stabilizer measurement circuits, illustrated by a weight-$10$ $X$ stabilizer measurement.  
The black CNOTs have targets on the 10 data qubits, collectively represented by a black wire.  In~(b-d), fault-tolerance is only guaranteed to distance three and Pauli corrections, or frame updates, are applied to the data based on the $Z$ basis measurements.
(a)~Shor's method uses $w+1$ ancillas and requires a fault-tolerantly prepared cat state. 
(b,c)~These methods use unverified cat states with subsequent error decoding, 
giving a deterministic circuit.
(d)~Our flag method prepares and unprepares an ancilla cat state while collecting the stabilizer.  Exponentially more flag patterns can thus be accessed for fault diagnosis.
}
\label{f:stabmeascomparison}
\vspace{-0.5cm}
\end{figure*}

We strive for low qubit overhead since quantum computers with limited qubits count resources preciously, and even minor improvements can free up extra qubits for other tasks. In topological codes where stabilizers are localized in space and are of  low weight, only a few flag qubits close to each stabilizer suffice to impart fault tolerance~\cite{YoderKim16trianglecodes, Chamberlandtriangularcodesflag2020, Chamberlandtopologicalsubsystemflag2020}.  It has also been shown that with adaptive control and quickly resetting qubits, only four ancillas are required for the universal fault-tolerant operation of some distance-three codes~\cite{ChaoReichardt17errorcorrection, ChaoReichardt18fewqubitcomputation}.  In this paper, we present a general fault-tolerant protocol that works for a stabilizer of any size.  If qubits are connected well enough, we show that only logarithmic overhead is required for fault-tolerant stabilizer measurement, an exponential space improvement over the previous linear overhead.

The general model of flag-based fault tolerance is displayed in \figref{f:genove}. Here a set of flag ancilla qubits monitor operations in a non-fault-tolerant circuit and when measured at the end, produce flag patterns that identify mid-circuit faults. Based on the observed flag pattern, a correction is applied to the data to minimize the spread of errors. As an example, \figref{f:syndromemeasurementd3slowresetw10} measures a stabilizer on 10 data qubits while tolerating one fault. The three colored qubits are the flags and the measured flag patterns each imply different corrections.  Also note that the sequence of flag patterns $100, 110, 111, 011, 001$ is a path on the hypercube and corresponds to the order of the flag CNOTs, e.g., between $100$ and $110$ a CNOT targets flag qubit~$2$.

In this paper, we restrict discussion to the measurement of individual stabilizers of a quantum code, as in Shor-style fault-tolerant stabilizer measurement~\cite{Shor96}. We do not consider measuring multiple stabilizers in parallel, as in Refs.~\cite{Steane, Knill05, huang2021}.  \fullfigref{f:stabmeascomparison} displays improvements made to Shor's method.  Note that Shor's method can tolerate any number of faults by increasing the fault tolerance of the cat state preparation.  The subsequent schemes forgo this property and are only fault-tolerant to distance three.  DiVincenzo and Aliferis first make the circuit deterministic by removing the need for cat state verification~\cite{DiVincenzoAliferis06slow}. This ensures that a circuit designer need not wait for a fault-tolerantly prepared cat state before measuring the stabilizer.  Subsequent improvements were made in Refs.~\cite{Stephens14colorcodeft, YoderKim16trianglecodes, ChaoReichardt18fewqubitcomputation} to reduce ancilla count by coupling each ancilla qubit to two data qubits instead of one.  

With our flag method, the ancilla cat state is prepared and unprepared while collecting the stabilizer.  As in \figref{f:syndromemeasurementd3slowresetw10}, an $X$ fault occurring anywhere on the $\ket +$ qubit may spread into the data, but will also leave its imprint on the flags. This is then measured out as a flag pattern. Due to the particular arrangement of the flag CNOTs, any fault that can spread to a data error of weight more than one triggers one of the five shown flag patterns. Each flag pattern then applies a unique correction that ensures that there is at most one data qubit in error.  This satisfies the condition for fault tolerance, which states that $k$ faults in a circuit should cause no more than $k$ qubits to have~errors. 

For the distance-three fault-tolerant measurement of a weight-$w$ stabilizer, we propose two methods based on the speed of qubit reset. With fast qubit reset, \thmref{t:fastresetd3syndromemeasurement}, only three flag ancillas are required in total, but each flag needs to be measured once per four data qubits.  If more flags are used in parallel, the number of accessible flag patterns grows exponentially and the number of measurements per ancilla converges to one. This is the regime of slow qubit reset, \thmref{t:syndromemeasurementd3slowreset}, which uses at most $\lceil \log_2 w \rceil$ flag ancillas measured only at the end.  Additionally, we show circuits for distance-five and distance-seven fault-tolerant stabilizer measurement in \appref{app:d5d7}.

\begin{table}
\caption{\label{f:resultsCSP}
Cat state size for different preparation methods that use $m$ ancilla qubit measurements.  
}
\begin{tabular}{ c @{\hspace{.25cm}} c }
\hline \hline
\textbf{Method} & \textbf{Cat state size~$w$} \\
\hline 
& \\[-.3cm]
 \textit{Deterministic} Correction& $w \leq 3 \, (2^m - 2m + 2)$ \\
(\thmref{t:catstated3}) & 
$\depth = (w - 1) + 2^{m - 2}$
\\[.15cm]
\textit{Adaptive} Correction &  $w \leq 3 \, (2^m - 2m + 3)$ \\
(\thmref{t:adaptiveslowresetd3}) & \\[.15cm]
 Error Detection &  $w \leq 3 \cdot 2^{m -1}$ \\
(\thmref{t:errordetcatd3}) & \\[.15cm]
\textit{Parallelized} Correction & $w = 2m = 2 \cdot 2^j , j\in \mathbb{N}$ \\
(\thmref{t:parallelcatd3s}) & $\depth = 2 + \log_2 w$\\
\hline \hline
& \\[-0.6cm]
\end{tabular}
\end{table}

\tabref{f:resultsCSP} contains bounds on the ancilla overhead for preparing weight-$w$ cat states fault-tolerantly to distance-three.  If the flag qubits can reset quickly, \thmref{t:catstated3} states that only one flag qubit is required and it needs to be reset and measured $m$ times.  Since the flag qubits operate independently, it is also possible to use $m$ flag qubits, with each one being measured once.  We further show how to use an adaptive circuit in \thmref{t:adaptiveslowresetd3} to marginally increase the number of flag patterns in use.

Appendices~\ref{s:d3ED} and~\ref{s:d3Par} contain two additional circuits for distance-three weight-$w$ cat state preparation. In \thmref{t:errordetcatd3}, we show how to use postselection to prepare cat states while tolerating \textit{two} faults.  Finally \thmref{t:parallelcatd3s} details how to create low-depth circuits for distance-three fault-tolerant cat state preparation, which may be useful in technologies with many qubits or long two-qubit gate~times. 

The rest of the paper is divided into three sections. \secref{s:flagseq} details the construction of the two paths on the hypercube that we use as flag sequences. \secref{s:d3syndmeas} describes how to use these sequences for distance-three fault-tolerant stabilizer measurement, and \secref{s:d3csp} deals with cat state preparation.

\section{Flag sequences}
\label{s:flagseq}

A flag pattern is a string of $1$s and $0$s that arises from measuring flag qubits.  A flag pattern with $a$ flags is a vertex of the $a$-dimensional hypercube~$\{0,1\}^a$.  We show how to construct two maximal-length paths through the hypercube.  
Between sequential flag patterns only one bit changes, which in the fault-tolerant circuit constructions below will correspond to a CNOT from the syndrome qubit to that flag qubit.

The first type of flag sequence just requires a maximal-length traversal of the $a$-dimensional hypercube. A simple choice is the Gray code~\cite{Gray1953, Gardner1986Doughnuts}.  

\begin{lemma} \label{t:graycode}
For $a \geq 1$, the Gray code creates a length-$2^a$ Hamming path in the $a$-dimensional hypercube~$\{0,1\}^a$.
\end{lemma}

\begin{proof}
We construct the sequence inductively.
For $a = 1$, use $0, 1$.  For $a > 1$, first run the sequence for $a-1$ with $0$s appended, then run it backwards with $1$s appended.  
\end{proof}

\noindent
For $a = 2$, e.g., the sequence is $00, 10, 11, 01$.  For $a = 3$, the sequence is $000, 100, 110, 010, 011, 111, 101, 001$.  

\smallskip

The second type of sequence is related to the degree of fault tolerance of the circuit.  By definition, fault tolerance to distance $d$ implies that for all $k \leq t =  \lfloor \frac{d-1}{2} \rfloor$, correlated errors of weight $k$ occur with $k$-th order probability.  For distance-three Calderbank-Shor-Steane (CSS) fault-tolerant syndrome measurement, any single fault should result in a data error with $X$ and $Z$ components having weight zero or one.  

In order to ensure that the circuit is distance-three fault-tolerant, we need to ensure that a measurement fault on any one ancilla qubit does not trigger corrections of weight greater than one.  Hence the second maximal-length sequence requires that there are no weight-one strings except at the start and end.  As shown in \figref{f:syndromemeasurementd3slowresetw10}, we may assign weight-one corrections to these two patterns, but for all the others, multi-qubit corrections~are~required.

\begin{figure}
\hspace{-0.25cm}
\includegraphics[width = 0.49\textwidth]{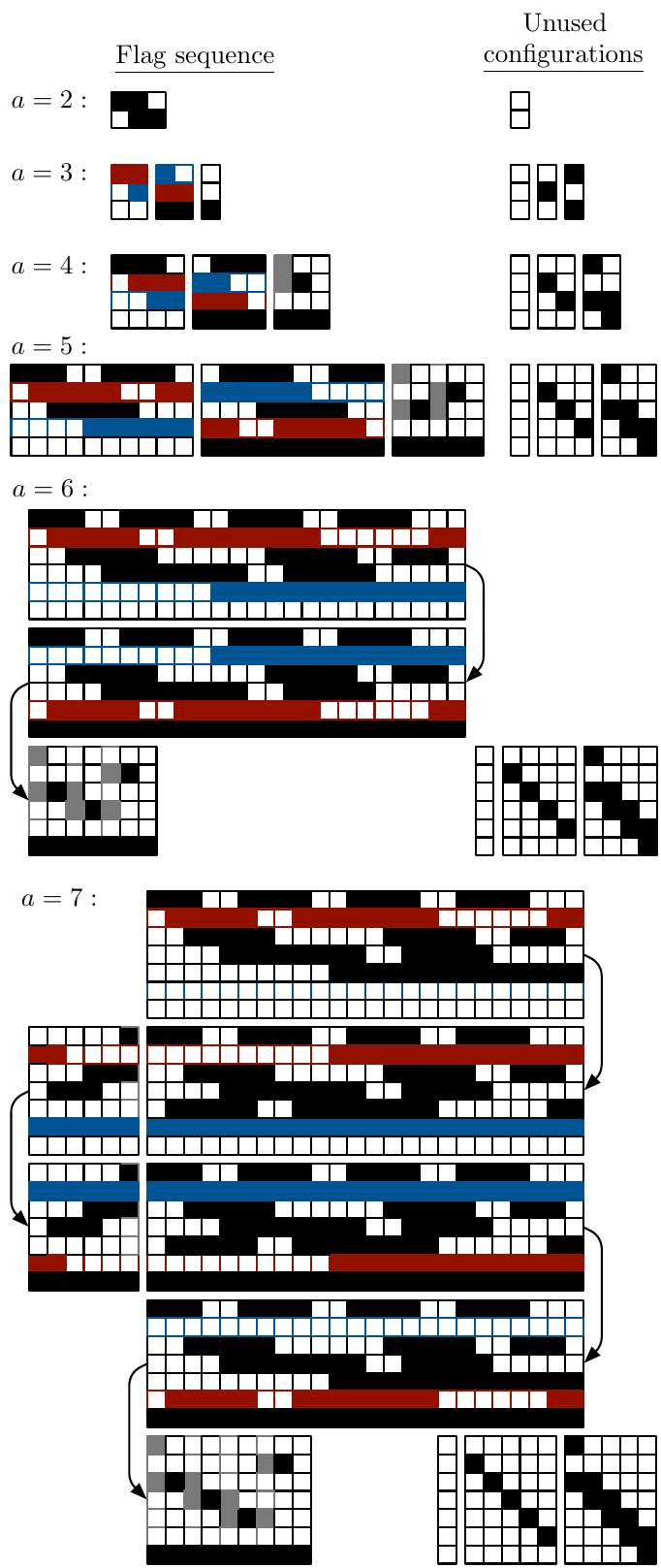}
\caption{
Flag sequences for distance-three fault-tolerant syndrome measurement, using $a$ flag qubits, each measured once (the slow reset model).  These sequences are walks through the $a$-dimensional hypercube, from $10^{a-1}$ to $0^{a-1}1$; passing through each vertex at most once and no other weight-one vertices.  Flag patterns are stacked vertically and ordered initially left to right, with solid and empty squares representing $1$ and $0$, respectively, e.g., \protect\raisebox{-.05cm}{\protect\includegraphics[scale=.4]{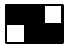}} represents $10, 11, 01$.} 
\label{f:slowresetdistance3flagsequences}
\end{figure}

\begin{lemma} \label{t:slowresetdistance3flagsequences}
For $a \geq 2$, in the $a$-dimensional hypercube $\{0,1\}^a$ there exists a path $v_1 = 10^{a-1}, \ldots, v_n = 0^{a-1}1$, with length $n = 2^a - 2 a + 3$, such that all $v_2, \ldots, v_{n-1}$ have weight at least two, and none repeat.  
\end{lemma}

\begin{proof}
\def\indicator{{1}}
Let $\indicator_S \in \{0,1\}^a$ denote the vertex that is $1$ exactly for indices in $S$.  
\fullfigref{f:slowresetdistance3flagsequences} illustrates the inductive construction.  For $a=2$, the sequence is the same as that in \lemref{t:graycode}. The base case of our inductive proof is with $a=3$, where the sequence is $100, 110, 111, 011, 001$. For $a > 3$, first run the previous sequence for $b = a - 1$ with $0$s added to the bottom, up to the second-to-last element 
(which for $b \geq 3$ is $\indicator_{\{2, b\}}$).  Then run the sequence backward, except with $1$s added to the bottom, and swapping coordinates $2$ and $a-1$ (the red and blue rows in the figure).  Finally, finish the sequence from $\indicator_{\{1, a\}}$ by walking through $\indicator_{\{3, a\}}, \indicator_{\{4, a\}}, \ldots, \indicator_{\{a - 2, a\}}, \indicator_{\{2, a\}}$, with the appropriate weight-three sequences $\indicator_{\{1, 3, a\}}, \indicator_{\{3, 4, a\}}, \ldots,$ $\indicator_{\{a - 1, a - 2, a\}}, \indicator_{\{a - 2, 2, a\}}$ (shown in gray) interposed.

To ensure that no vertex is visited more than once, one need only check that the last $2a - 5$ sequences are distinct from those that came before.  For this, one can track by induction the $2a - 3$ hypercube vertices that are not visited by each walk: $0^a$, the $a - 2$ weight-one strings $\indicator_2, \ldots, \indicator_{a - 1}$, and the $a - 2$ weight-two strings $\indicator_{\{1,3\}}, \indicator_{\{3,4\}}, \indicator_{\{4,5\}}, \ldots, \indicator_{\{a - 1, a\}}$.  
Thus, the sequence has total length $2^a - (2 a - 3)$.  
\end{proof}

The length $2^a - 2a + 3$ is maximal.  
This follows since there are $2^{a-1} - a$ vertices with odd weight more than one, and vertices must alternate odd and even~weights.  

\section{Distance-three stabilizer measurement}
\label{s:d3syndmeas}

In this section, we outline two protocols for distance-three CSS fault-tolerant stabilizer measurement.  They differ based on the speed of qubit measurement and reset.

For $w \in \{4, 5, 6\}$, flag-fault-tolerant circuits are constructed the same way regardless of qubit reset speed.  We show in \figref{f:slowresetd3w6stabilizermeasurement} that for $w = 6$, only two flag qubits are required.  Lower-weight stabilizers can be measured by removing data CNOTs and making appropriate changes to the Pauli corrections.  For $7 \leq w \leq 10$, the different constructions yield the same circuits.  It is only for $w > 10$ that the effects of qubit reset speed are~pronounced.

\subsection{Fast reset}

\begin{theorem} \label{t:fastresetd3syndromemeasurement}
If qubits can be measured and reset quickly, then for any $w$, four ancilla qubits are sufficient to measure the syndrome of $X^{\otimes w}$, CSS fault-tolerantly to distance three.  Moreover, the number of measurements needed is $\lceil \tfrac{w + 2}{4} \rceil + 1$.  
\end{theorem}

\begin{proof}
For $w \in \{ 4, 5, 6\}$, the circuit using two flag ancillas is shown in \figref{f:slowresetd3w6stabilizermeasurement}.  It runs through a sequence of three flag patterns and a multi-qubit correction is only applied for the flag pattern $11$.  For $w > 6$, the general construction is shown in \figref{f:syndromemeasurementd3fastreset}.  Each repetition of the highlighted region adds the $X$ parity of four more data qubits, while measuring and quickly reinitializing one flag qubit.  In terms of the number of measurements~$m$, the construction achieves up to $w = 4 \, (m - 1) - 2$.  
It is fault-tolerant because $X$ faults on the control wire cause flag patterns of alternating weights two or three, that localize the fault to three possible consecutive locations along the control wire: before, between or after two CNOT gates.  The appropriate correction, ensuring distance-three fault tolerance, is for a fault between the CNOT gates.  
\end{proof}

\begin{figure}
\centering
\subfigure[\label{f:slowresetd3w6stabilizermeasurement}]{\includegraphics[width = 0.35\textwidth]{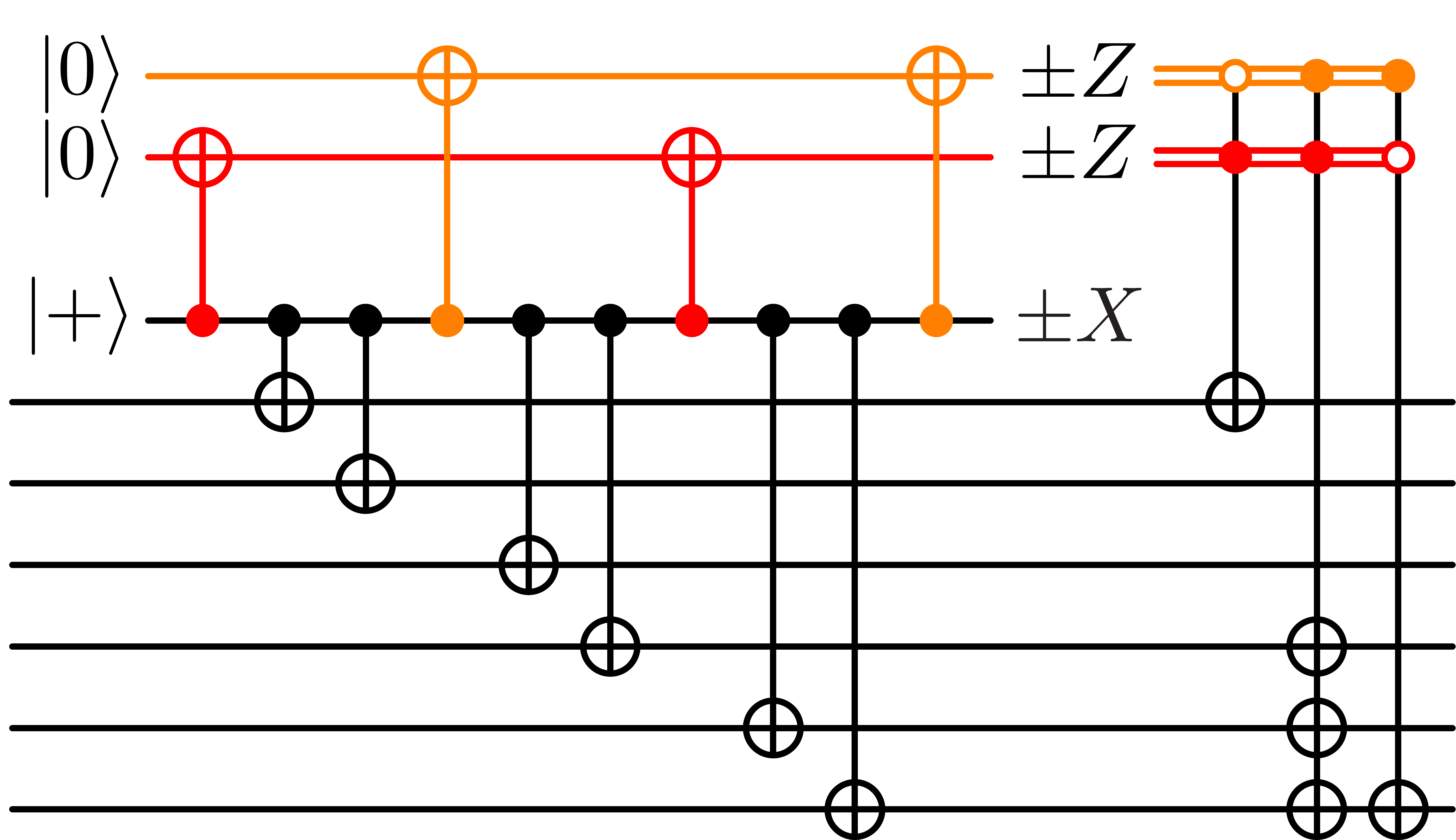}}
\subfigure[\label{f:slowresetd3w6CSP}]{\includegraphics[width = 0.35\textwidth]{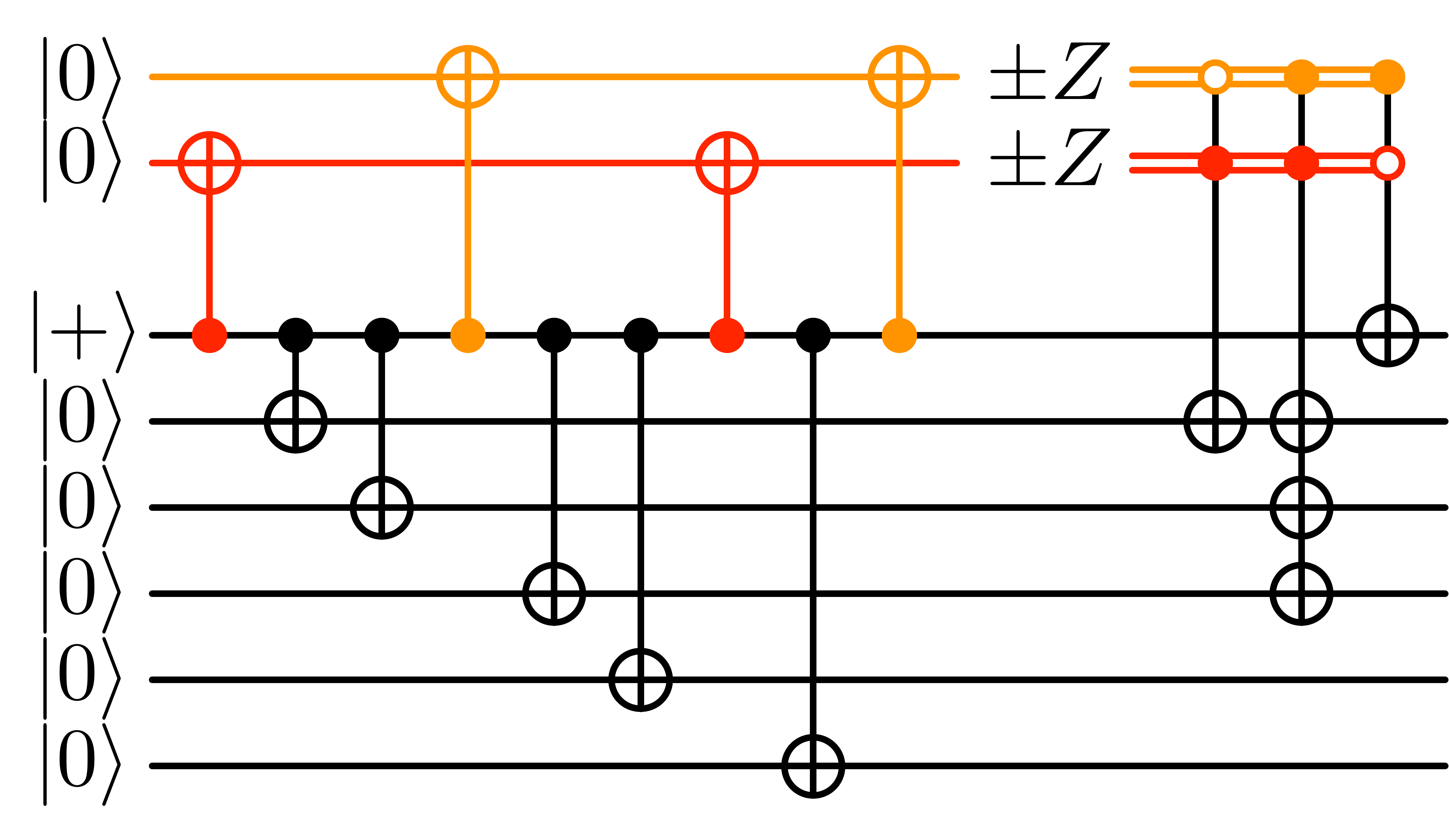}}
\caption{
(a) Circuit to measure an $X^{\otimes 6}$ stabilizer, CSS fault-tolerant to distance three.  
(b) Circuit to prepare a six-qubit cat state, fault-tolerant to distance three.  
}
\label{f:distance3w6}
\end{figure}

\begin{figure*}
\centering
\includegraphics[scale=.75]{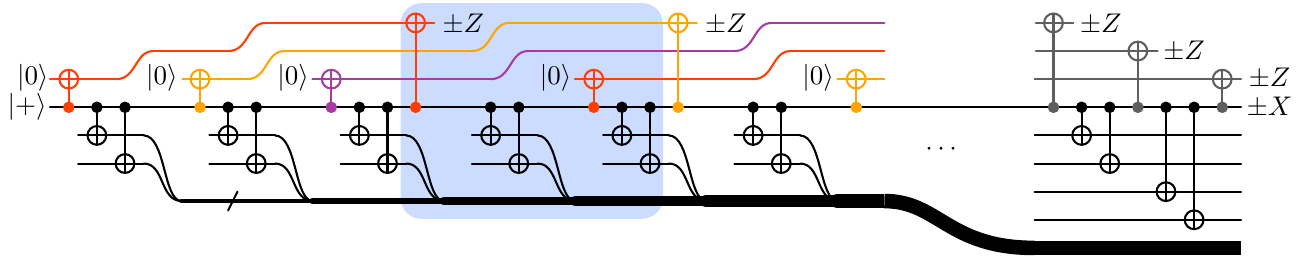}
\caption{
Distance-three fault-tolerant syndrome bit measurement only needs three flag qubits.
The highlighted region can be repeated to fit the weight of the stabilizer being measured.}
\label{f:syndromemeasurementd3fastreset}
\end{figure*}

\begin{figure}
\centering
\includegraphics[width=.32\textwidth]{oneflagunprotected}\vspace{0.2cm}
\includegraphics[width=.28\textwidth]{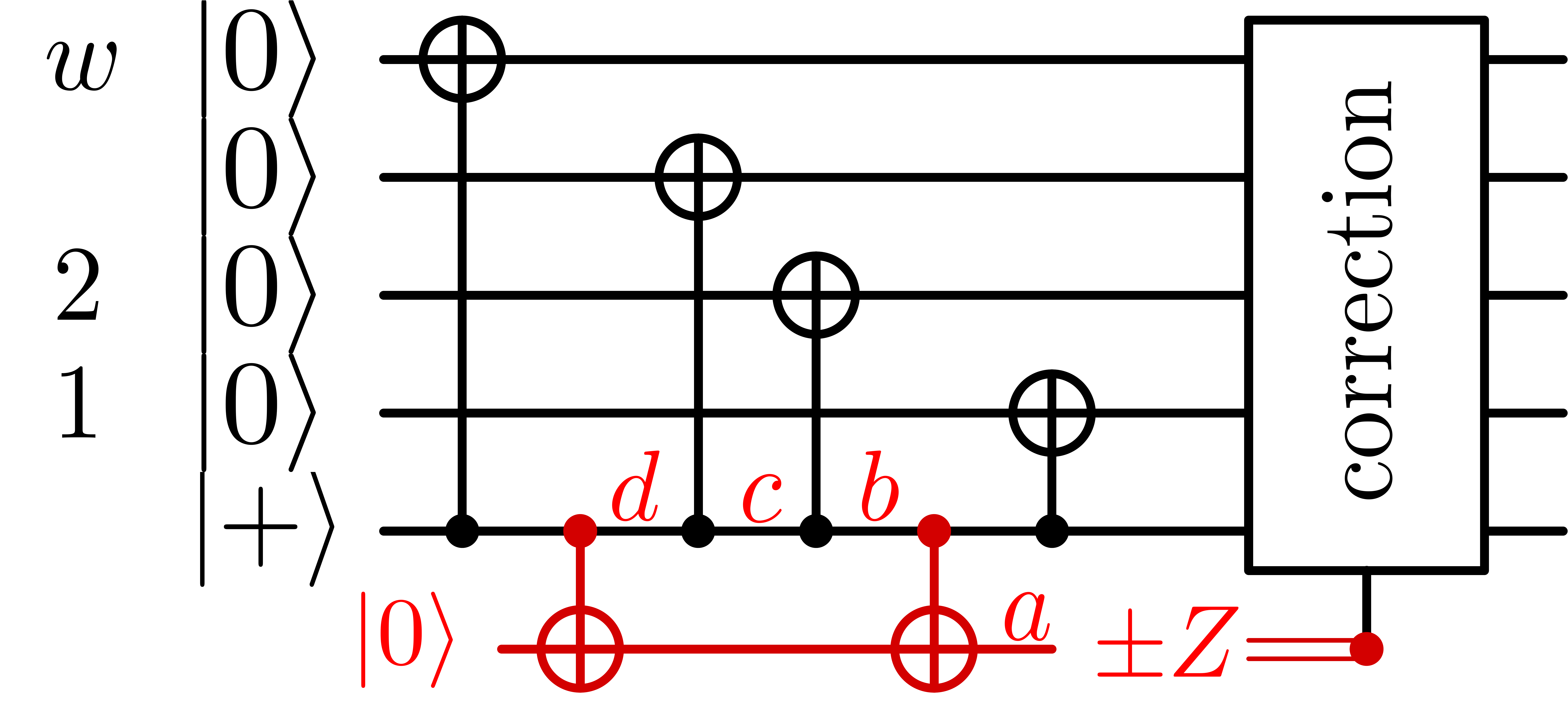} 
\caption{Distance-three error correction is not possible with one flag qubit.  Either (top) the control wire is unprotected at some point ${\color{red} \star}$, from which an $X$ fault can propagate to an error of weight at least two; or (bottom) faults at ${\color{red} a}$, ${\color{red} b}$, ${\color{red} c}$, ${\color{red} d}$, causing respective errors $I$, $X_1$, $X_1 X_2$, $X_w$ have no consistent~correction.} \label{f:oneflagprotection}
\end{figure}

\thmref{t:fastresetd3syndromemeasurement} may be optimal; it does not appear to be possible to use fewer than three flag qubits.  With just one flag qubit, one can detect that an error has occurred, but not where.  As illustrated in \figref{f:oneflagprotection}, either the control wire is unprotected at some point or for $w \geq 4$ there is no consistent correction rule.  

By a similar argument, two flag qubits are not enough.  Any correction based on a single flag can have weight at most one, since the flag measurement itself could be faulty.  However, if at some point in the middle the control wire is protected by just a single flag, a weight-one correction will not suffice.  On the other hand, if both flags are used to protect the control wire across the entire sequence of CNOT gates, we are unable to locate faults well enough to correct~them.  

\pagebreak 

We remark that this construction can also be used to prepare a $w$-qubit cat state fault-tolerantly to distance three. The conversion follows three steps: 1.\ Remove one data qubit.  2.\ Initialize the data qubits as $\ket 0$.  3.\ Remove the syndrome ancilla measurement, so as to retain it in the support of the stabilizer.  An example of this conversion is shown for $w = 6$ in \figref{f:slowresetd3w6CSP}.  In \secref{s:d3csp}, we will give a better protocol that uses just one ancilla qubit.

\subsection{Slow reset}
\label{subsec:stabmeasslow}

\begin{theorem} \label{t:syndromemeasurementd3slowreset}
The syndrome of $X^{\otimes w}$ can be measured CSS fault-tolerantly to distance three using $m \geq 3$ measurements, provided that 
\begin{equation*}
w \leq 2 \, (2^{m - 1} - 2 (m - 1) + 3) \, .
\end{equation*}
\end{theorem}

\begin{proof}
Two examples are shown in \figref{f:slowresetd3w6stabilizermeasurement}, for $w = 6$, and \figref{f:syndromemeasurementd3slowresetw10}, for $w = 10$.  
As in these figures, in general we collect the syndrome two qubits at a time into a syndrome qubit that is initialized as $\ket +$.  Between each of these pairs of CNOT gates, a CNOT is applied from the syndrome qubit into one of $m - 1$ flag qubits.  This leads to a sequence of flag patterns, e.g., $100, 110, 111, 011, 001$ for the $w = 10$ example.  Based on the observed flag pattern, a correction is applied as if an $X$ fault had occurred between the corresponding pair of flag CNOT gates.  

Observe that the flag sequence changes one bit at a time; it can be thought of as a path on the hypercube.  It begins and ends with weight-one patterns, but otherwise the patterns all have weight at least two.  This is important for distance-three fault tolerance because a fault could affect the flags, and only the first and last data corrections have weight one.  Also, the flag patterns along the sequence are distinct, so each is associated with only one correction.  The theorem then just follows using the flag sequence construction in \lemref{t:slowresetdistance3flagsequences}.  
\end{proof}

Note that the approach of \thmref{t:syndromemeasurementd3slowreset}, with slow reset, is different from the fast reset case of \thmref{t:fastresetd3syndromemeasurement}, in that a flag qubit is active and able to detect faults in more than one region of the circuit.

\section{Distance-three cat state preparation}
\label{s:d3csp}

Next we turn to the question of distance-three fault-tolerant preparation of cat states.  For preparing a two- or three-qubit cat state, any preparation circuit is automatically fault-tolerant, because every error has weight zero or one.  For example, on three qubits $XXI \sim IIX$, since $XXX$ is a stabilizer.  Fault tolerance becomes interesting for preparing cat states on $w \geq 4$ qubits.  

The ideas of Theorems~\ref{t:fastresetd3syndromemeasurement} and~\ref{t:syndromemeasurementd3slowreset} can also be applied to cat state preparation.  For example, just as in \figref{f:distance3w6} a circuit for measuring $X^{\otimes 6}$ with three ancilla qubits corresponds to a circuit to prepare a six-qubit cat state with two ancillas, similarly adapting the construction of \thmref{t:syndromemeasurementd3slowreset} allows preparing a $2 (2^a - 2 a + 3)$-qubit cat state using $a$ ancilla qubits each measured once.  However, we can do better.  

\begin{theorem} \label{t:catstated3}
For $m \geq 2$, one ancilla qubit, measured $m$ times, is sufficient to prepare a cat state on $w$ qubits fault-tolerantly to distance three, for
\begin{equation*}
w \leq 3 \, \big( 2^m - 2 m + 2 \big)\,.
\end{equation*}
\end{theorem}

\noindent
Let $[m] = \{1, 2, \ldots, m\}$ and $X_S = \prod_{j \in S} X_j$.

\begin{proof}[Proof of \thmref{t:catstated3}]
\fullfigref{f:catstated3} illustrates our construction for the cases $m = 3$ and $m = 4$.  In general, we prepare a $w$-qubit cat state using CNOT gates from the first qubit, so that the possible $X$ errors from a single fault are $\identity, X_1, X_{[2]}, X_{[3]}, \ldots$.  
We then compute parities of subsets of the qubits into the ancillas, following the flag sequence from \lemref{t:slowresetdistance3flagsequences} and \figref{f:slowresetdistance3flagsequences}.  Although for clarity \figref{f:catstated3} shows the $m$ parity checks being made in parallel, they can also be made sequentially with just one ancilla qubit.  

\begin{figure*}
\centering
\subfigure[\label{f:catstatew12}]{\includegraphics[scale=.75]{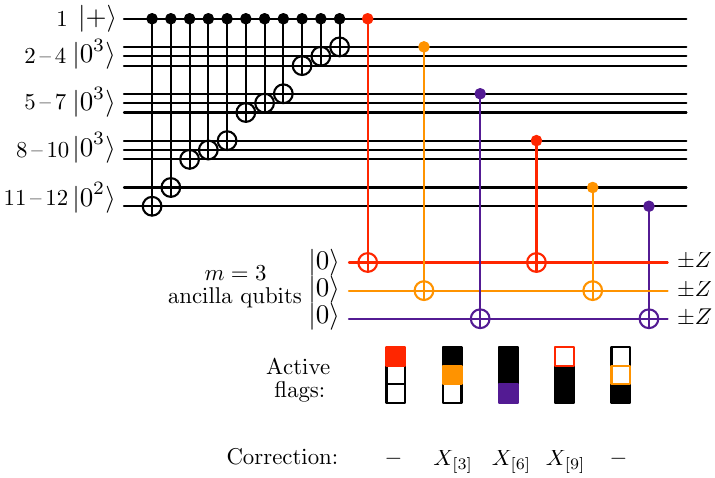}}
\subfigure[]{\includegraphics[scale=.75]{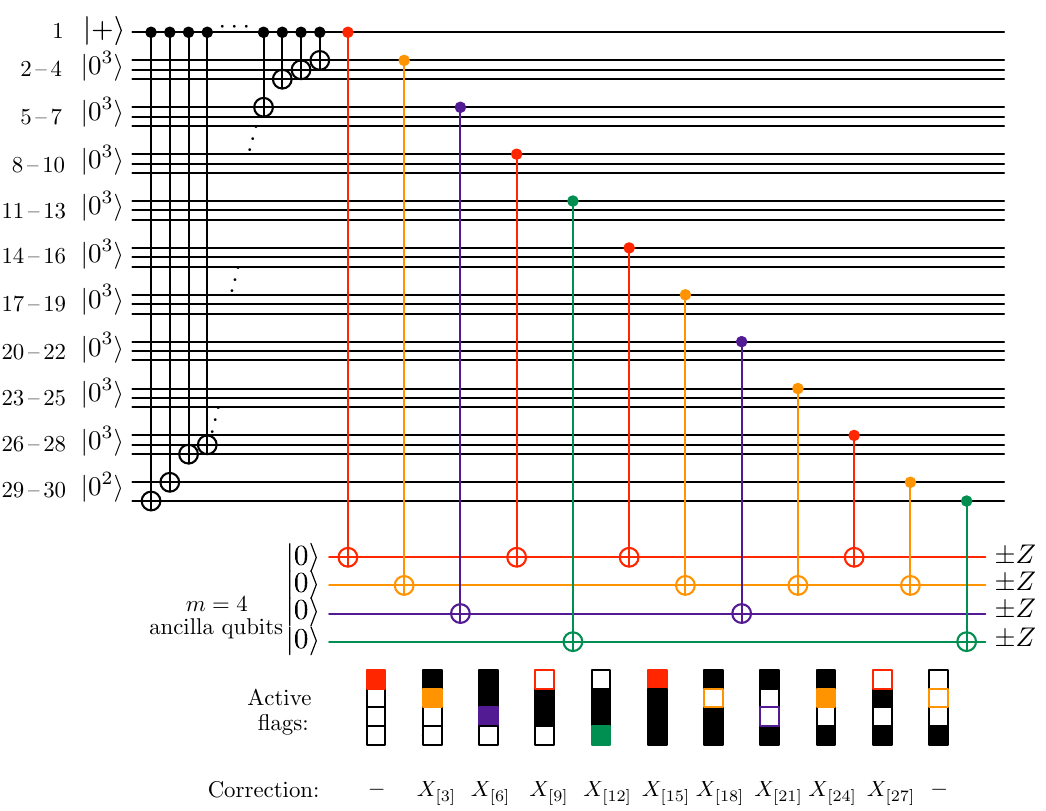}}
\caption{Distance-three fault-tolerant cat state preparation circuits.  Note that, with fast reset, only one ancilla qubit is required.} \label{f:catstated3}
\end{figure*}

With the given correction rules, errors due to single faults are corrected up to possibly a weight-one remainder.  (For example, in \figref{f:catstatew12}, errors $X_{[5]}$, $X_{[6]}$ and $X_{[7]}$ all result in the parity checks $111$, for which the correction $X_{[6]}$ is applied.)  The circuit also tolerates faults within the parity-check sub-circuit, because a single fault here can flip at most one parity, and no correction is applied for the weight-one patterns.  
\end{proof}

By this method, the cat state is prepared in depth $w - 1$.  The depth of the parity check circuit increases exponentially as $2^{m - 2}$ for $m \geq 3$ if we consider slow reset ($a = m$).  This is evident from the flag sequences in \figref{f:slowresetdistance3flagsequences} as the maximum number of times any flag bit is switched.  The total depth of the circuit is then $(w - 1) + 2^{m - 2}$.  

Note that the construction from \thmref{t:catstated3} does not help for syndrome measurement, because the parity checks would in general become entangled with the data. 

We can do slightly better if we allow an \emph{adaptive} circuit, in which the parity checks are chosen based on the outcome of a flag qubit measurement.  For example, \fullfigref{f:adaptiveslowresetd3w15} gives a circuit to prepare a $15$-qubit cat state using $m = 3$ measurements.  Here, the result of measuring the red ancilla determines how the other two ancillas are used.  

\begin{figure*}
\centering
\includegraphics[width=.85\textwidth]{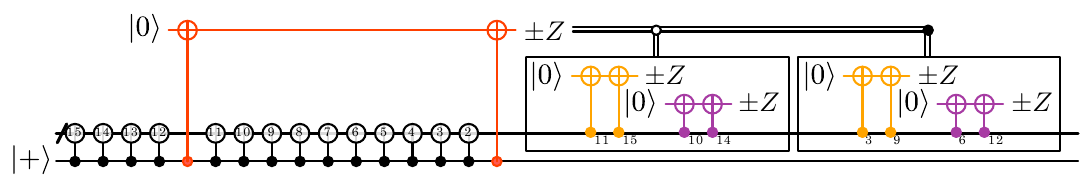}
\caption{Circuit to prepare a $15$-qubit cat state by adaptive error correction, fault-tolerant to distance three.  Labels on the thick black wire indicate which data qubit in the block is being addressed as the control or target of the CNOT.  If a fault occurs while preparing the cat state on the $\ket +$ qubit, it is partially localized by the red flag ancilla.  The measurement result of this flag then determines a set of parity checks to completely localize a possible fault.  After all the ancilla qubits have been measured, corrections are applied based on \tabref{f:adaptiveslowresetd3w15corrections}.  
}
\label{f:adaptiveslowresetd3w15}
\end{figure*}


\begin{theorem} \label{t:adaptiveslowresetd3}
Using an adaptive circuit, for $m \geq 2$, one ancilla qubit, measured $m$ times, can be used to prepare a cat state on $w$ qubits fault-tolerantly to distance three, for
\begin{equation*}
w \leq 3 \, \big( 2^m - 2 m + 3 \big)\, .
\end{equation*}
\end{theorem}

\begin{figure*}
\centering
\subfigure[$w = 15, a = 3$]{\includegraphics[scale=.8]{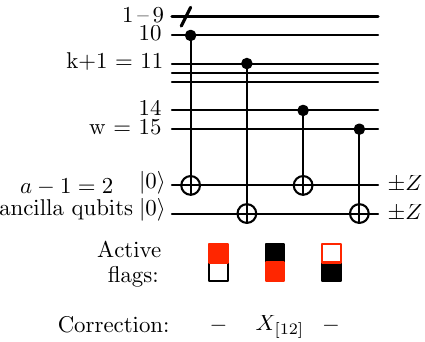}}\hspace{1cm}
\subfigure[$w = 33, a = 4$]{\includegraphics[scale=.8]{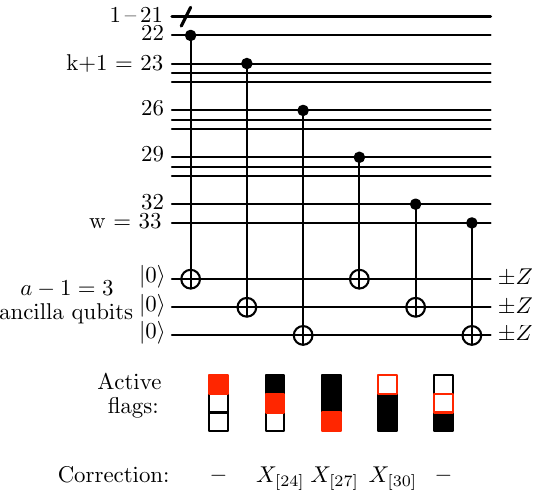}}
\subfigure[$w = 75, a = 5$]{\includegraphics[scale=.8]{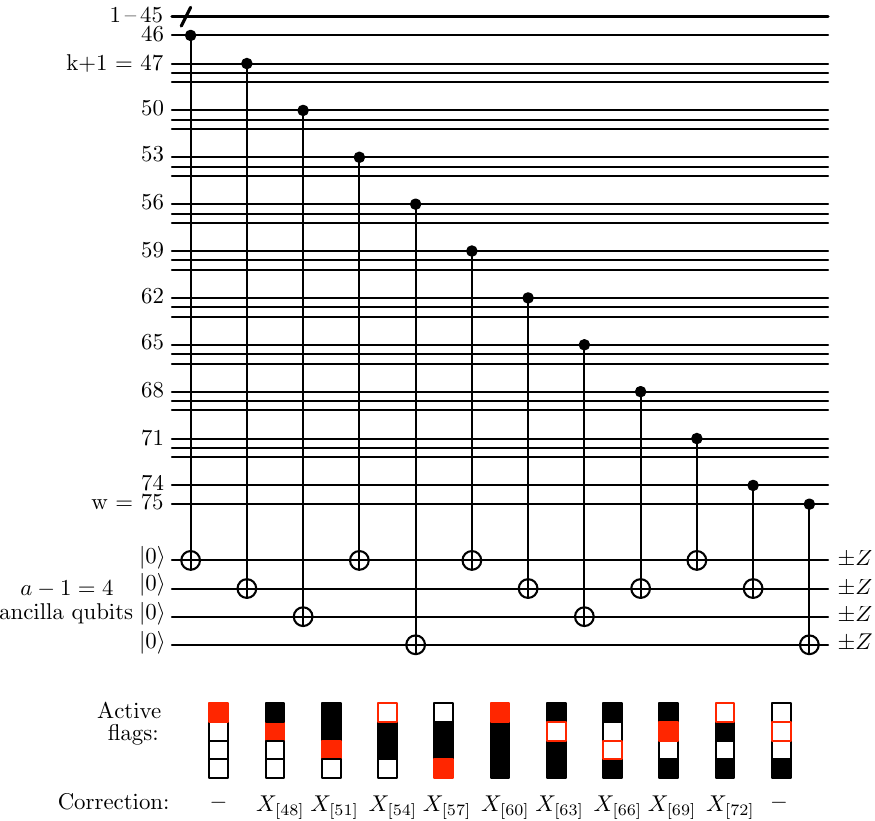}}
\caption{If the red ancilla flag in \figref{f:adaptiveslowresetd3w15} is not triggered, these circuits are used to find and correct a possible error.  The flag sequences (from \figref{f:slowresetdistance3flagsequences}) and corresponding corrections are listed at the bottom.   Note that these sequences are nonadaptive, and can be used either with $a$ ancilla qubits in a slow reset model, or with just one ancilla qubit in a fast reset model, since all the CNOT gates commute.  
} \label{f:adaptived3noflagcorrections}
\end{figure*}

\begin{proof}
Our construction will follow the same basic structure as the circuit in \figref{f:adaptiveslowresetd3w15}.  
Prepare the $w$ data qubits as $\ket{{+} 0^{w-1}}$, then apply $\CNOT_{1,w}, \CNOT_{1,w-1},$ $ \ldots, \CNOT_{1,2}$ to get a cat state.  Let $k = 3 ( 2^{ m - 1 }) - 2 $.  Just before $\CNOT_{1,k+1}$ and just after $\CNOT_{1,2}$, apply CNOTs into the first ancilla qubit, the red qubit in \figref{f:adaptiveslowresetd3w15}, and measure it.  

The remainder of the circuit depends on the measurement result.  If it is~$1$, then a fault has been detected.  The error on the cat state can be one of 
\begin{equation*}
\identity, X_1, X_{[2]}, \;\;\; X_{[3]}, X_{[4]}, X_{[5]}, \;\;\; \ldots, \;\;\; X_{[k-1]}, X_{[k]}, X_{[k+1]} \, .
\end{equation*}
The correction procedure needs to determine in which of the above $1 + \tfrac{k-1}{3}$ groups-of-three the error lies; then for any error in $\{X_{[3 j]}, X_{[3 j + 1]}, X_{[3 j + 2]}\}$ the correction $X_{[3 j + 1]}$ works.  Perhaps the easiest way to locate the error is by binary search using the Gray code in \lemref{t:graycode}, e.g., by computing parities between qubits $3 j$ for~$j \in \{ 1, 2, \ldots, 1 + \tfrac{k-1}{3} \}$.  Since the measurement of the red ancilla could have been incorrect, it is important that the all-$0$s outcome of the binary search correspond to the $\identity, X_1, X_{[2]}$ error triple, as in~\tabref{f:adaptiveslowresetd3w15corrections}.  Using $m - 1$ measurements, we can search $2^{m-1}$ possibilities, which indeed is $1 + \tfrac{k-1}{3}$.  (The search circuit can also be made nonadaptive, as in \figref{f:adaptiveslowresetd3w15}.)  

Next consider the case that the first measurement result is~$0$, so no fault has been detected.  The error on the cat state can be one of $X_{[k+1]}, X_{[k+2]}, \ldots, X_{[w]} \sim \identity$.  We again use the remaining $m - 1$ ancilla qubits to measure parities of subsets of cat state qubits.  Since there is no guarantee of a fault having occurred yet, we use flag sequences from \lemref{t:slowresetdistance3flagsequences}, where the length of the weight-at-least-two flag sequence is $J = 2^{m - 1} - 2 ( m - 1 ) + 1$.  The parity checks are now done between qubits $\{ k, k + 1 + 3 j, k + 2 + 3 J \} $ for $ j \in \{ 0, 1, \ldots , J \}$, as shown in \figref{f:adaptived3noflagcorrections} and \tabref{f:adaptiveslowresetd3w15corrections}.  We do not allow weight-one flag patterns to be able to correct any errors since they may also be triggered by a measurement fault on any one of the data qubits involved in the parity check.

\begin{table*}
\vspace{0.2cm} 
\caption{Possible data errors and associated corrections for the different observed flag patterns in \figref{f:adaptiveslowresetd3w15}.  $[m] = \{1, 2, \ldots, m\}$.} \label{f:adaptiveslowresetd3w15corrections}
\centering
\begin{tabular}{c c c c c c}
\hline
\hline
 & & & & & \\[-0.33cm]
{\color{red} Red flag} & \multicolumn{2}{c}{Parity checks} & & Possible errors & Correction \\
\hline
 & & & & \\[-0.3cm]
 $1$ & ${\color{Tangerine} 3 \oplus 9}$ & ${\color{Purple} 6 \oplus 12}$ & & &  \\
 \cline{2-3}
 & & & & & \\[-0.3cm]
 & $0$ & $0$ & &  $\identity , X_1, X_{[2]}$ & $X_1$ \\ 
 & $1$ & $0$ & &  $X_{[3]}, X_{[4]}, X_{[5]}$ & $X_{[4]}$ \\ 
 & $1$ & $1$ & &  $X_{[6]}, X_{[7]}, X_{[8]}$ & $X_{[7]}$ \\ 
 & $0$ & $1$ & &  $X_{[9]}, X_{[10]}, X_{[11]}$ & $X_{[10]}$ \\[0.25cm]
 $0$ & ${\color{Tangerine} 11 \oplus 15}$ & ${\color{Purple} 10 \oplus 14}$ & & &  \\
 \cline{2-3}
 & & & & & \\[-0.3cm]
 & $0$ & $0$ & & $\identity$ & None \\ 
 & $1$ & $0$ & &  $X_{11}, X_{15}, X_{[14]}$ & None \\
 & $1$ & $1$ & & $X_{[11]}, X_{[12]}, X_{[13]}$ & $X_{[12]}$ \\
 & $0$ & $1$ & & $ X_{10}, X_{14}$ & None \\
  & & & & & \\
\hline 
\hline
\end{tabular}
\end{table*}

\begin{table*}
\centering
\vspace{0.2cm} 
\caption{Space and time costs for measuring a weight-$w$ stabilizer using different distance-three fault-tolerant stabilizer measurement circuits. In the following, all the logarithms are base $2$. The flag method requires the fewest ancillas and has low depth, allowing for the smallest cost when computing $\# \text{ancillas} \times \text{depth}$. }
\begin{tabular}{c c c c}
\hline
\hline
Protocol & Ancillas & Depth & Ancillas$\times$Depth \\
\hline
Shor & $w+1$ & $w/2 + 3$ & $O(0.5 {\color{red}w^2})$ \\
Shor-Par & $5w/4$ & $3\log w - 1 $ & $O(3.75 {\color{red}w \log w})$ \\
DA & $w$ & $2w - 1$ & $O(2 {\color{red}w^2})$ \\
Compressed DA & $w/2$ & $3w/2 - 2$ & $O(0.75 {\color{red}w^2})$ \\
Flag & $\log w +1$  & $3w/2 + O(1)$ & $O(1.5 {\color{red}w \log w})$ \\[0.25cm]
Not fault-tolerant & $1$ & $w$ & $O({\color{red}w})$ \\
\hline
\hline
\end{tabular}
\label{f:resourcereqs}
\end{table*}

Consolidating, we are allowed up to $3 J + 1$ CNOTs before the red ancilla is initialized, and up to $k$ CNOTs in the monitored region of the red ancilla. In total we can create a cat state on up to 
\begin{equation*}
w \leq 3 J + k + 2 = 3 \, \big( 2^m - 2 m + 3 \big)
\end{equation*}
qubits, with $m$ total measurements.
\end{proof}

We also tested protocols where multiple flags are used for the initial partial localization of a fault (in place of the red flag qubit).  We found no improvement to our bounds on ancilla overhead.  It appears that ancillas are better used in the parity checks than for partial fault~localization. 

\vspace{-0.2cm} 
\section{Simulation and space-time cost}

We count the circuit depth and number of ancillas used in our distance-three fault-tolerant stabilizer measurement circuits. Parallelization can substantially reduce circuit depth. \tabref{f:resourcereqs} compares our flag method for measuring a weight-$w$ stabilizer to the earlier methods in \figref{f:stabmeascomparison}. Also considered is a parallelized Shor method, in which the initial cat state is prepared in logarithmic depth, with $w/4$ extra ancilla qubits for postselection checks. The Shor methods must pass the postselection checks, and so they are non-deterministic protocols. \tabref{f:resourcereqs} shows the best case, where all the checks pass. Note that the flag and parallelized Shor methods both have $\text{space} \times \text{depth}$ cost  scaling as $O(w \log w)$, with the leading coefficient in favor of the flag method.

Using a standard depolarizing noise model, we simulate noisy versions of the different circuits to determine statistics of the weight-one and weight-two errors. Specifically:
\vspace{0.25cm} 
\begin{itemize}[leftmargin=*]
\item With probability $p$, the preparation of $\ket 0$ is replaced by $\ket 1$ and vice versa---similarly for $\ket +$ and $\ket -$.
\item With probability $p$, an $X$ or $Z$ measurement has its outcome flipped.
\item With probability $p$, a one-qubit gate is followed by a Pauli error drawn uniformly at random from $\{ X, Y, Z\}$.
\item With probability $p$, the two-qubit CNOT gate is followed by a two-qubit Pauli error drawn uniformly at random from $\{ I, X, Y, Z\}^{\otimes 2} \setminus \{I \otimes I \}$.
\end{itemize}
\vspace{0.25cm} 
There are no errors on idle resting qubits. 
\vspace{0.2cm} 

\begin{figure*}
\centering
\includegraphics[width=.97\textwidth]{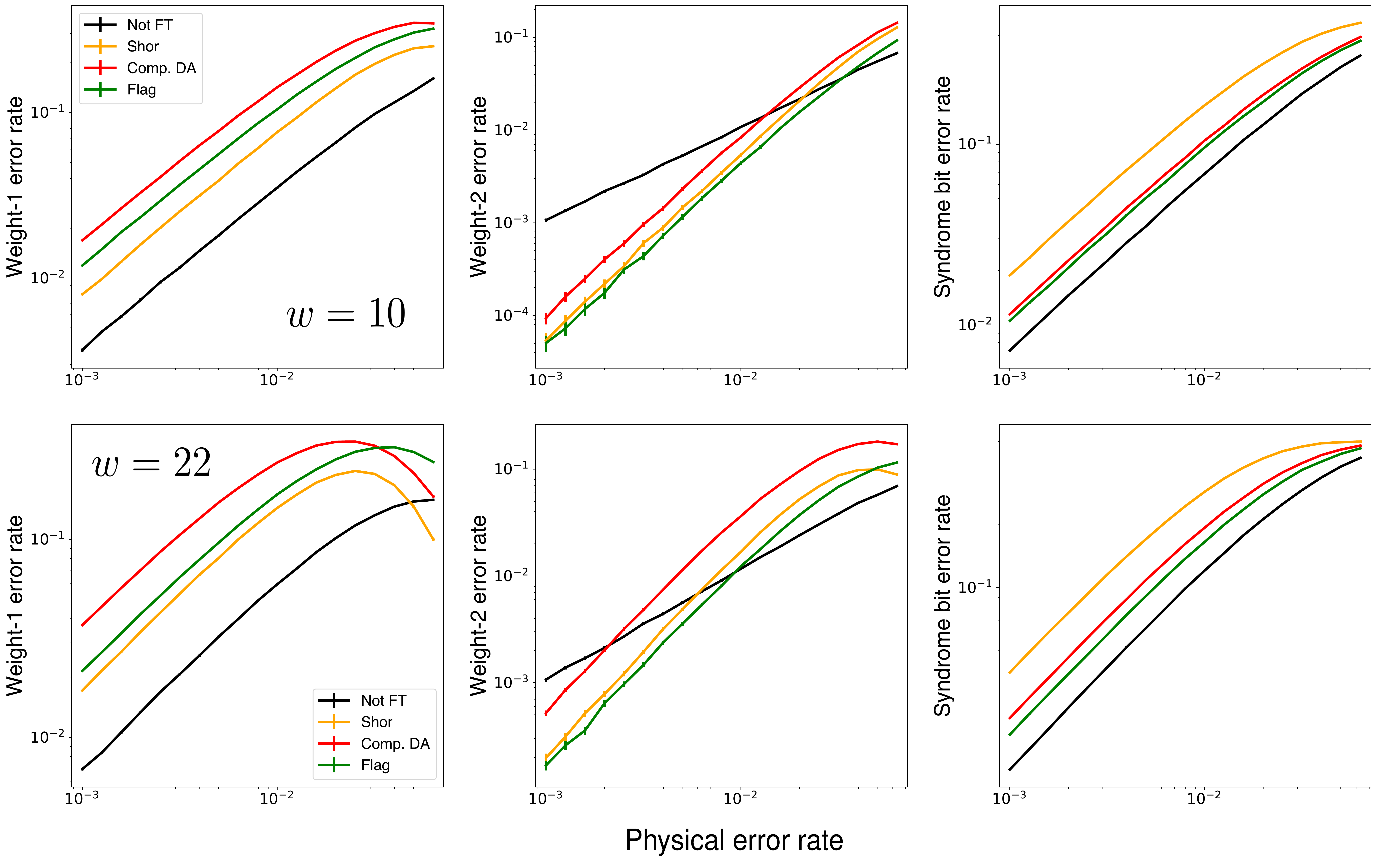}
\caption{Simulation of the noisy measurement of an $X ^{\otimes 10}$ and $X ^{\otimes 22}$ stabilizer at physical error rate $p \in \{ 10^{-3}, 10^{-2} \}$ using different distance-three fault-tolerant circuits: Shor-style, compressed Divincenzo-Aliferis, and the flag method of \secref{subsec:stabmeasslow}. In the first and second column of graphs, we show the rate of weight-one and weight-two data errors due to these circuits, with $99 \%$ error bars. In the third column, we show the rate at which the measured syndrome bit is wrong. 
} \label{f:sims}
\end{figure*}

\fullfigref{f:sims} shows the rates of weight-one errors and weight-two errors for different input physical error rates $p$. The rate of weight-one errors is lowest in the non-fault-tolerant circuit, since it contains the fewest locations for faults. The Shor method has a lower weight-one error rate than the flag method, but among the deterministic fault-tolerant methods, the flag method performs the best. For larger stabilizers ($w=22$), the curves for the Shor method and the flag method are closer, implying that the difference in the rate of weight-one errors between the Shor method and the flag method is reduced. 

As expected, the rate of weight-two errors of the three fault-tolerant protocols scales quadratically with $p$, allowing for a lower probability of weight-two errors below a pseudothreshold (physical error rate below which a fault-tolerant method achieves lower weight-two error rate than the non-fault-tolerant method). Notice that the flag method has the highest pseudothreshold. Moreover, as the stabilizer weight is increased, the pseudothreshold of the flag method decreases slower than those of the other fault-tolerant methods. Asymptotically, the flag method admits the highest pseudothreshold for weight-two errors, but incurs more weight-one errors than the probabilistic Shor method. Additionally, we compute the rate of errors on the syndrome bit, as this determines how much fault tolerance will be needed to correct faulty syndrome information~\cite{Delf22}. The rate of faulty syndrome bits is lowest when using the flag method for fault tolerance.

\section{Conclusion} 

In this paper, we optimize the overhead of distance-three fault-tolerant stabilizer measurement and cat state preparation.  If the circuit on $w$ qubits must tolerate one fault, we show that $\sim \! \log w$ extra qubits are sufficient.  
We detail the construction of a maximal-length path on the hypercube and show that, compared to previous flag schemes, it allows using the extra flag qubits more efficiently to catch and distinguish faults.

We describe two circuits for stabilizer measurement based on the speed of ancilla qubit reset.  With slow reset, a weight-$w$ stabilizer can be measured fault-tolerantly to distance three using only $\lceil \log_2 w \rceil$ flag qubits for fault tolerance.  With fast reset, only three flag qubits are required, but the number of times they are measured and reset grows as $w/4$. 

In our circuits for fault-tolerant cat state preparation we check for errors after the cat state is non-fault-tolerantly prepared. We show, using a deterministic and an adaptive circuit, that the overhead for fault tolerance can be as low as logarithmic in the size of the cat state. In fact, only one flag qubit suffices, as long as it can reset quickly.  

We now turn to further improvements.  The circuits detailed in this paper are only fault-tolerant to distance three.  
However, flags can be used to effect fault tolerance to arbitrary distance~\cite{Chamberland2018flagfaulttolerant, chao2019flag}, and it is open to develop higher-distance fault-tolerant stabilizer measurement circuits with low overhead.  

From the perspective of stabilizer algebra, a cat state is a CSS ancilla state.  A future avenue of research is to extend these flag techniques to fault-tolerantly and deterministically prepare more complex CSS ancilla states. 

In order to execute the circuits in this paper, one qubit needs to be connected to all the other qubits used.  This is concerning for architectures with limited connectivity, such as superconducting qubits.  But by mixing flag and transversal gate concepts for fault tolerance, it is possible to construct stabilizer measurement circuits that can measure arbitrarily large stabilizers using only local interactions, fault-tolerantly.  

\section{Acknowledgements}

We would like to thank Rui Chao, Sourav Kundu and Zhang Jiang for insightful conversations.  Research supported by Google and by MURI Grant FA9550-18-1-0161.  This material is based on work supported by the U.S. Department of Energy, Office of Science, National Quantum Information Science Research Centers, Quantum Systems~Accelerator.

\appendix

\section{Postselected cat state preparation tolerating two faults}
\label{s:d3ED}

Shor's method for fault-tolerant stabilizer measurement relies on the fault-tolerant preparation of a cat state by postselection.  In \figref{f:Shorcard}, the cat state is prepared fault-tolerantly to distance-two; it detects one fault.  For postselected distance-three fault tolerance, any one or two faults in the circuit must result in an error of weight at most one or two respectively, else the state must be rejected. In \figref{f:d3 error detection} we show how to prepare a weight-$12$ cat state fault-tolerantly to distance three---detecting up to two faults.


\begin{theorem} \label{t:errordetcatd3}
One ancilla qubit measured $m \geq 2$ times, can be used to prepare a cat state on $w$ qubits fault-tolerantly to distance three, detecting up to two faults,~for
\begin{equation*}
 w \leq 3 \cdot 2^{m - 1} \, .
 \enspace 
\end{equation*}
\end{theorem}

\begin{proof}
We explain the proof using the circuit in \figref{f:d3 error detection}.  The circuit passes with acceptable weight-one or weight-two errors when all the flag qubits are measured as $0$.  If one $X$ fault occurs on the $\ket +$ qubit during the preparation of the cat state, it may spread to a data error of weight~$> 1$. However the red flag qubit is triggered and the fault is detected.  If two $X$ faults occur on the $\ket +$ qubit, the red flag qubit may not catch it, yet a data error of weight~$> 2$ can exist on the cat state.  Since this scenario only arises from two faults, it suffices to check the parities between every third qubit of the cat state, as an error on two consecutive qubits is acceptable.  Higher-weight errors, such as the weight-seven error $X_2X_3 \mathellipsis X_8$ in \figref{f:d3 error detection} may not be detected by parity checks that have an even number of erroneous qubits. However these errors are always caught by other parity checks.

To check for errors of weight greater than two, we perform parity checks similar to that in \thmref{t:catstated3}. Instead of the flag sequence from \lemref{t:slowresetdistance3flagsequences}, the Gray code from \lemref{t:graycode} is used. Now the parities are computed between qubits $3j - 1$ for $j \in \{1, 2, \ldots, 2^{m - 1}\}$.  The first and the last qubits are not checked for errors and so with $m$ flags, the maximum cat state weight achieved is $3 \cdot 2^{m - 1}$.  
\end{proof}

\begin{figure}
\centering
\includegraphics[scale=.75]{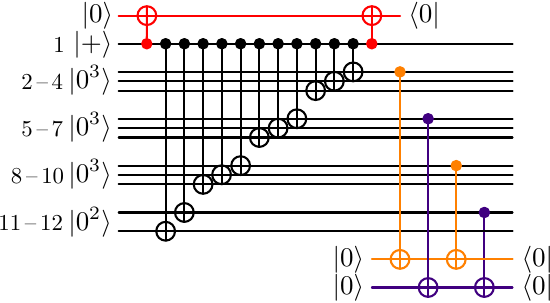}
\caption{Two-error-detecting fault-tolerant circuit for the preparation of a weight-$12$ cat state. The state is only accepted when all flag qubits are measured as $0$.  Note that with fast reset, only one ancilla qubit is required.}
\label{f:d3 error detection}
\end{figure}

\section{Parallelized distance-three cat state preparation}
\label{s:d3Par}

So far, we have focused on fault-tolerant preparation circuits with depth linear in the cat state weight.  In this section, we detail how to prepare cat states fault-tolerantly in logarithmic depth.

\begin{figure}
\subfigure[]{\includegraphics[scale=.15]{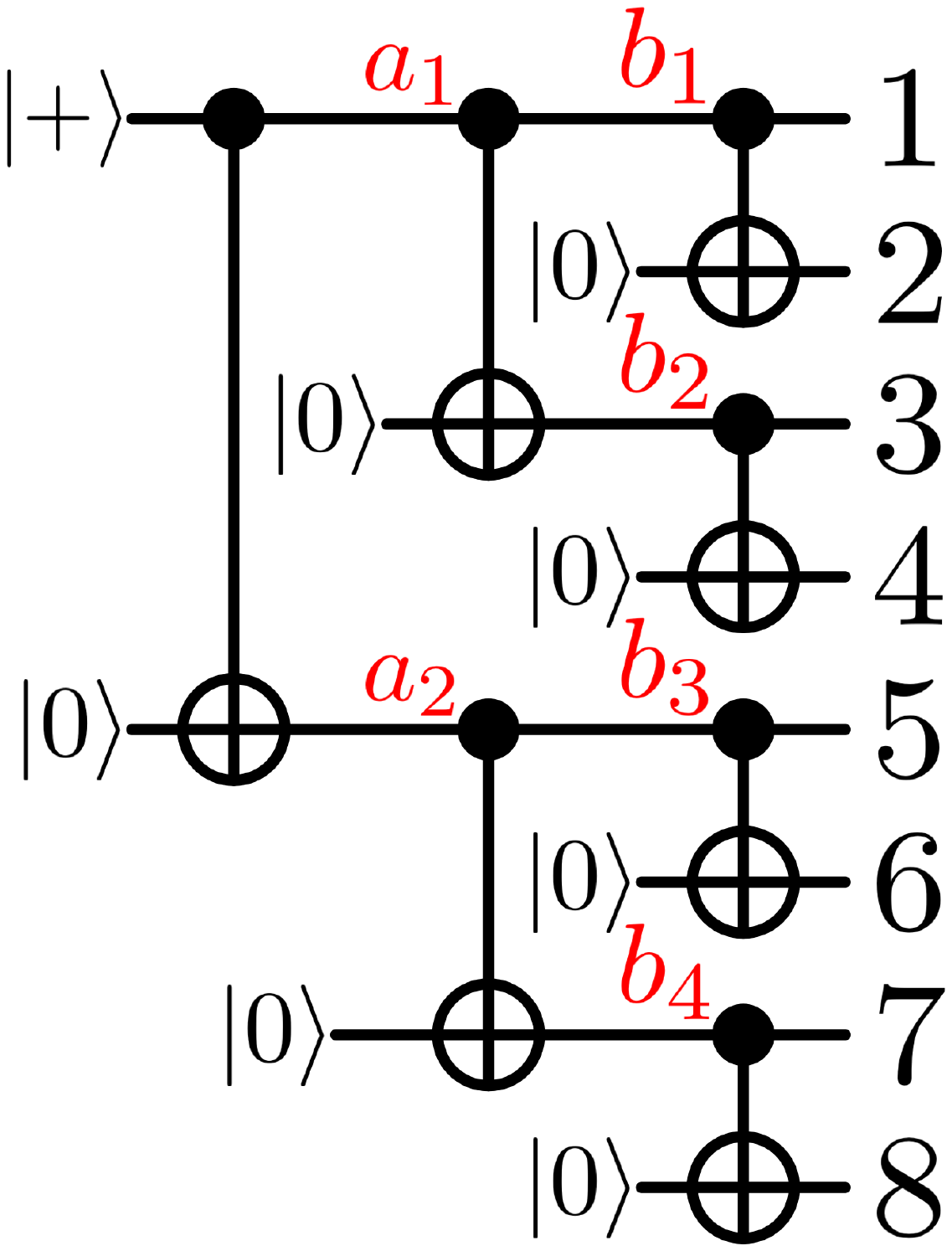}\label{f:parallelcsp8}} 
\subfigure[]{\includegraphics[scale=.15]{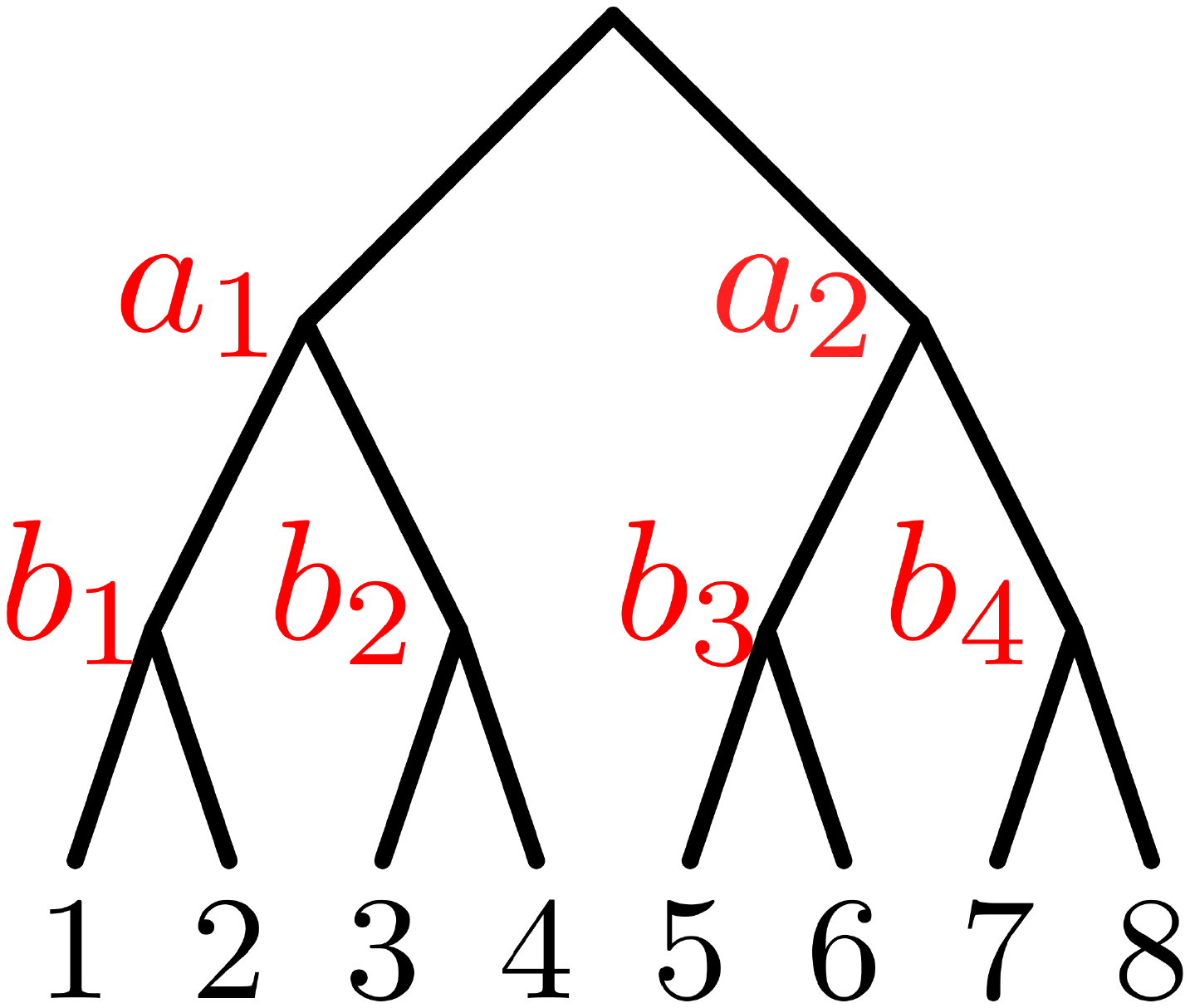}\label{f:parallelcsp8graph}}
\caption{(a)~Logarithmic-depth preparation of an eight-qubit cat state shows there are six possible locations for $X$ faults that create errors of weight at least two.  Parity checks need to be chosen to find corrections that leave the cat state with error of weight less than two. (b)~The circuit on the left can be represented as a graph, where a CNOT gate is represented by the splitting of an edge.} 
\end{figure}

In~\figref{f:parallelcsp8} an eight-qubit cat state is prepared in three rounds of CNOT gates.  There are six locations (marked in red) where an $X$ fault may cause an error of weight at least two.  These faults result in data errors with a different structure from the linear-depth protocols of \secref{s:d3csp}, hence different parity checks are required.  It is simpler to determine these parity checks if the circuit is viewed as a binary tree, as in~\figref{f:parallelcsp8graph}.  Here time flows down and every CNOT onto a fresh $\ket 0$ qubit is denoted by the splitting of an edge.  An $X$ fault at a marked location results in an $X$ error on all the leaf nodes directly under the location.  Note that a fault at the root cannot cause a bad error.  

We use only two-qubit parity checks, however larger parity checks may be used at the expense of increased depth.  If a parity check checks qubit $x$, it provides information on whether a fault occurred anywhere in the lineage: $l(x) = \{x, \pa(x), \pa(\pa(x)), \ldots , \ro\} $.  Therefore, if a parity check $(x,y)$ is triggered, a fault at one of the locations $ l(x) \cup l(y) \cup \{\spam \}$ has occurred, where $\{\spam \}$ is the set of faults during state preparation or measurement of the parity-check qubit. 

Using the parity checks $(1,5), (2,7), (3,6), (4,8)$, it is possible to separate the five distinct weight at least two errors (since the error due to $a_1$ and $a_2$ is the same up to the cat state's $X^{\otimes w}$ stabilizer) into distinct triggered flag~patterns:

\begin{center}
\begin{tabular}{c @{\hspace{0.5cm}} c c c c}
 & $(1,5)$ & $(2,7)$ & $(3,6)$ & $(4,8)$\\
 \hline 
 $a_1$, $a_2$ & $\bullet$ &  $\bullet$ &  $\bullet$ & $\bullet$ \\
$b_1$ &  $\bullet$ &  $\bullet$ & $\circ$ & $\circ$\\
$b_2$ & $\circ$ & $\circ$ & $\bullet$  &  $\bullet$\\
$b_3$ &  $\bullet$ & $\circ$ &  $\bullet$ & $\circ$\\
$b_4$ & $\circ$ &  $\bullet$ & $\circ$ & $\bullet$ \\
\end{tabular}
\end{center}

Note that a fault at any of the above locations requires a multi-qubit data correction.  We ensure that each of them is detected by at least two parity checks, as one faulty parity check must not induce corrections of weight greater than one.

\begin{theorem} \label{t:parallelcatd3s}
Using parallelized circuits, a $w$-qubit cat state can be prepared fault-tolerantly to distance three using $\frac{w}{2}$ parity checks, where $\frac{w}{2}= 2^j, \, j \in \mathbb{N}$.  The depth of the circuit is $2+ \log_2 w$.
\end{theorem}

\begin{proof}

For parity check $i \in \{1,2, \mathellipsis , \frac{w}{4} \}$, the cat state qubits checked are $(i, \frac{w}{2} + 2i - 1)$.  For the remaining parity checks  $i \in \{ \frac{w}{4} +1, \frac{w}{4} +2, \mathellipsis ,  \frac{w}{2} \}$, the qubits checked are $(i, 2i)$.  As in \figref{f:parallelcsp8graph}, faults at the ${\color{red}a}$ level (depth-one) locations trigger all the parity checks, since each parity check is executed on one cat state qubit from the first half, and one from the second.  The correction $X^{\otimes w/2}$ on either half of the qubits works for both faults as $(X^{\otimes w/2} \otimes \Id^{\otimes w/2})( \Id^{\otimes w/2} \otimes X^{\otimes w/2}) = X^{\otimes w}$ is a stabilizer of the cat state.  Faults at the ${\color{red}b}$ level (depth-two) trigger distinct sets of $\frac{w}{2^2}$ parity checks, where the correction is on all the leaf nodes under the uniquely identified fault.  The same holds for faults at depth-$k$, which trigger distinct sets of $\frac{w}{2^k}$ parity checks.

One faulty parity check leads to a weight-one flag pattern, for which we do not apply corrections, as the error is restricted to at most one cat state qubit.  
\end{proof}

\section{Distance-five and distance-seven fault-tolerant stabilizer measurement}
\label{app:d5d7}

Distance-five fault tolerance is interesting for stabilizers of weight $w \geq 6$.  For $w \in \{ 6, 7, 8\}$, the circuits in \figref{f:syndromemeasurementd5} with seven ancilla qubits are distance-five fault-tolerant.  We present a general method to construct stabilizer measurement circuits for arbitrary $w$ in \figref{f:d5arbwsyndmeas}.  By computer simulation, we verify the fault tolerance of this construction for $w$ up to $90$ qubits.  The general construction proceeds as follows. First, five flag qubits are activated. For each additional flag that is needed, the gates in the shaded blue region are applied.  These gates deactivate an existing flag and activate a new flag.  Finally, when no additional flags are needed, the flags are deactivated in the order $\{ 2,4,1,5,3\}$. $1$ denotes the flag that has been active for the longest time and $5$, the flag that has been active for the shortest time.  To ensure faults are correctly flagged, it is necessary to ensure there is asymmetry between the order in which flags are activated and the order in which they are deactivated. This is in contrast to the distance-three DiVincenzo-Aliferis method in \figref{f:CompDAcard}, where both the orders are symmetric.

In \figref{f:d5arbwsyndmeas}, the thick black line indicates the $w$-qubit register of data qubits that are in the support of the stabilizer.  Data CNOT gates (in black) are applied to qubits $\{ w, w-1, \mathellipsis , 1\}$ after every flag CNOT (in red).  The last data CNOT must be placed either before the third-last or second-last flag CNOT.  The addition of another data CNOT gate before the last flag CNOT results in uncorrectable errors. 

If there are $a$ ancilla qubits, one can measure a weight-$(2 a - 5)$ or weight-$(2 a - 4)$ stabilizer.  Hence a weight-$w$ CSS stabilizer may be fault-tolerantly measured to distance-five, for $w \leq 2a - 4$.  Note that at most five flag qubits are active at any instant.  Hence with fast qubit reset, one only requires five flag ancillas and one syndrome ancilla to measure an arbitrary weight stabilizer fault-tolerantly to distance-five.

For distance-seven fault-tolerance, we detail changes to the spacing between data CNOT gates and generalize the order in which flag ancillas are activated and deactivated.  We show how to construct circuits for stabilizer of arbitrary weight $w$ by first discussing a circuit for a weight-$17$ stabilizer, shown in \figref{f:d7w17}. We chose $w=17$ since the circuit is non-trivial and its construction encompasses all the tricks needed to construct circuits for arbitrary weight.  In general, compared to \figref{f:d5arbwsyndmeas}, the number of ancilla CNOT gates between data CNOT gates is doubled, except in the center of the circuit, where it is tripled for the length of four data CNOT gates.  For odd $w$, the number of ancilla CNOTs between the $w-1$ subsequent pairs of data CNOT gates is the sequence $\{ (\lceil \frac{w-6}{2} \rceil  \text{ 2's} ), 3, 3, 3, 3 , (\lfloor \frac{w-6}{2} \rfloor \text{ 2's}), 1\}$, as shown in \figref{f:d7w17}. For even $w$, the sequence is $\{ ( \frac{w-6}{2} \text{ 2's} ), 3, 3, 3, 3 , (\frac{w-6}{2} \text{ 2's}), 1\}$. Note that, as shown in \figref{f:d7w17}, one additional ancilla CNOT gate is required at the start.

Next we comment on the order in which ancilla qubits are deactivated as flags. Similar to the distance-five case, after initially activating seven flags, a flag is deactivated to activate a new flag qubit. An active group of seven flags is closed in the order $\{ 2,4,6,1,3,5,7 \}$. As these seven flags are closed, seven new flags are simultaneously opened. The process repeats unitl there are exactly seven remaining flags to close. These last seven flags are also closed in the same order $\{ 2,4,6,1,3,5,7 \}$. In \figref{f:d7w17}, flag ancillas are shown in alternating colors to highlight the order that flags are activated and deactivated. Distance-seven fault-tolerance was verified by computer simulation for stabilizer weight up to $32$. The number of flag ancillas needed to measure a weight-$w$ stabilizer is $w+1$.  

The techniques described in this section may also be used to develop resource-efficient circuits that are fault-tolerant to higher distance.

\begin{figure*}
\centering
\subfigure[$w = 6$, $a = 7$]{\includegraphics[width=.325\textwidth]{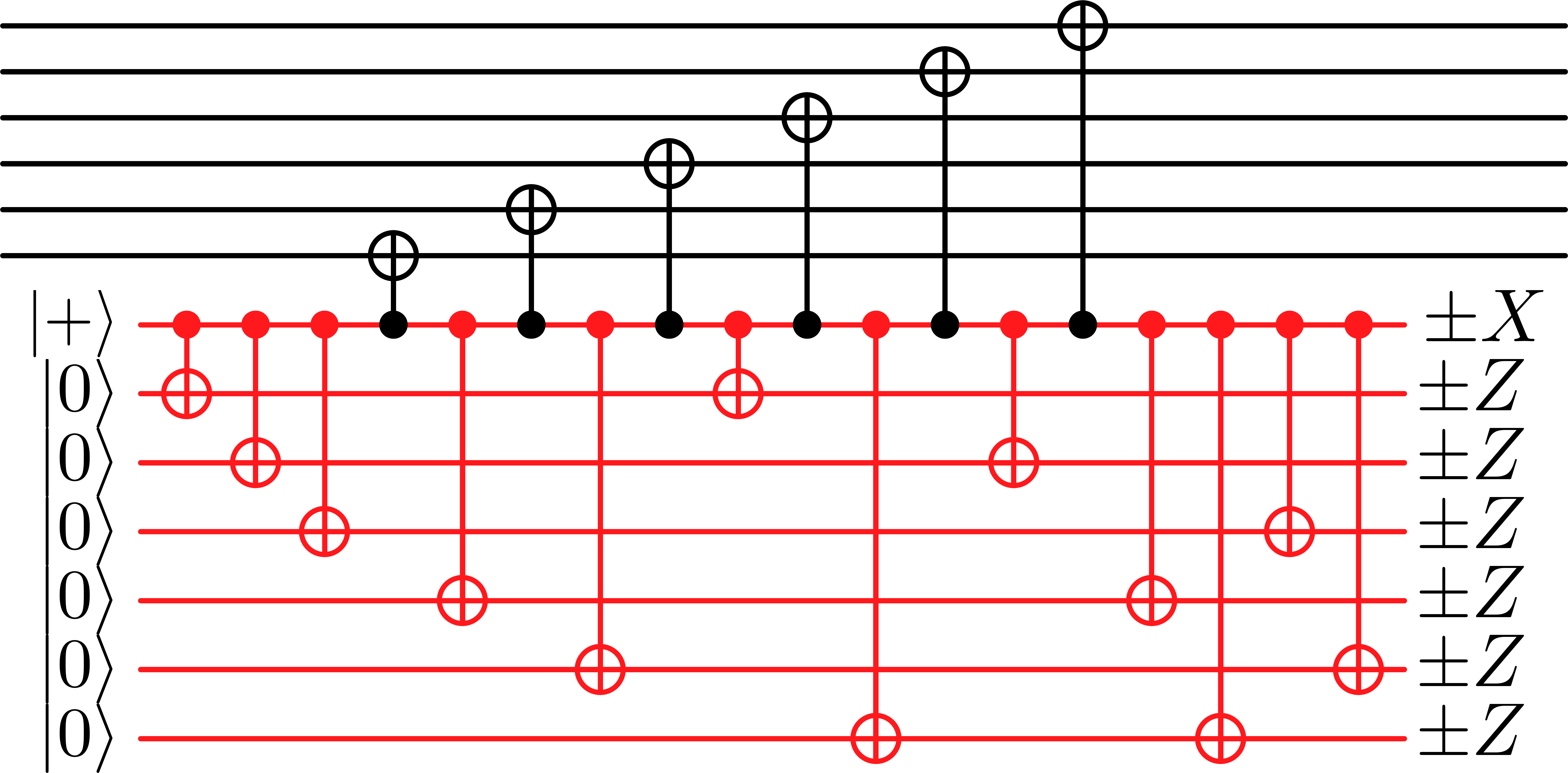}}     
\subfigure[$w = 7$, $a = 7$]{\includegraphics[width=.325\textwidth]{dist5w7}}
\subfigure[$w = 8$, $a = 7$]{\includegraphics[width=.325\textwidth]{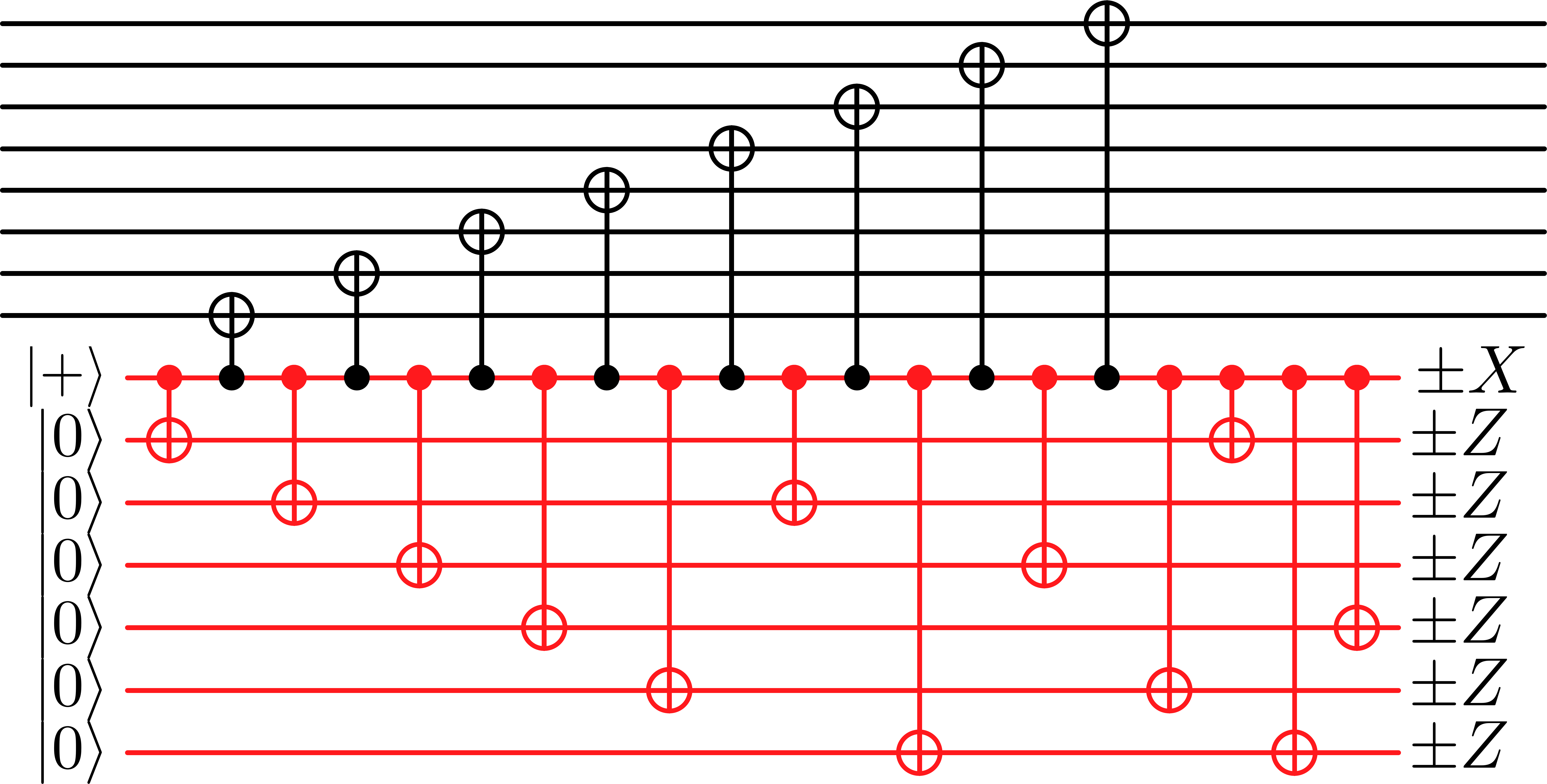}} 
\caption{Distance-five CSS stabilizer measurement with slow qubit reset for $w \in \{ 6, 7, 8\}$.  Red wires indicate syndrome and flag qubits.} 
\label{f:syndromemeasurementd5}
\end{figure*}

\begin{figure*}
\centering
\includegraphics[width=.62\textwidth]{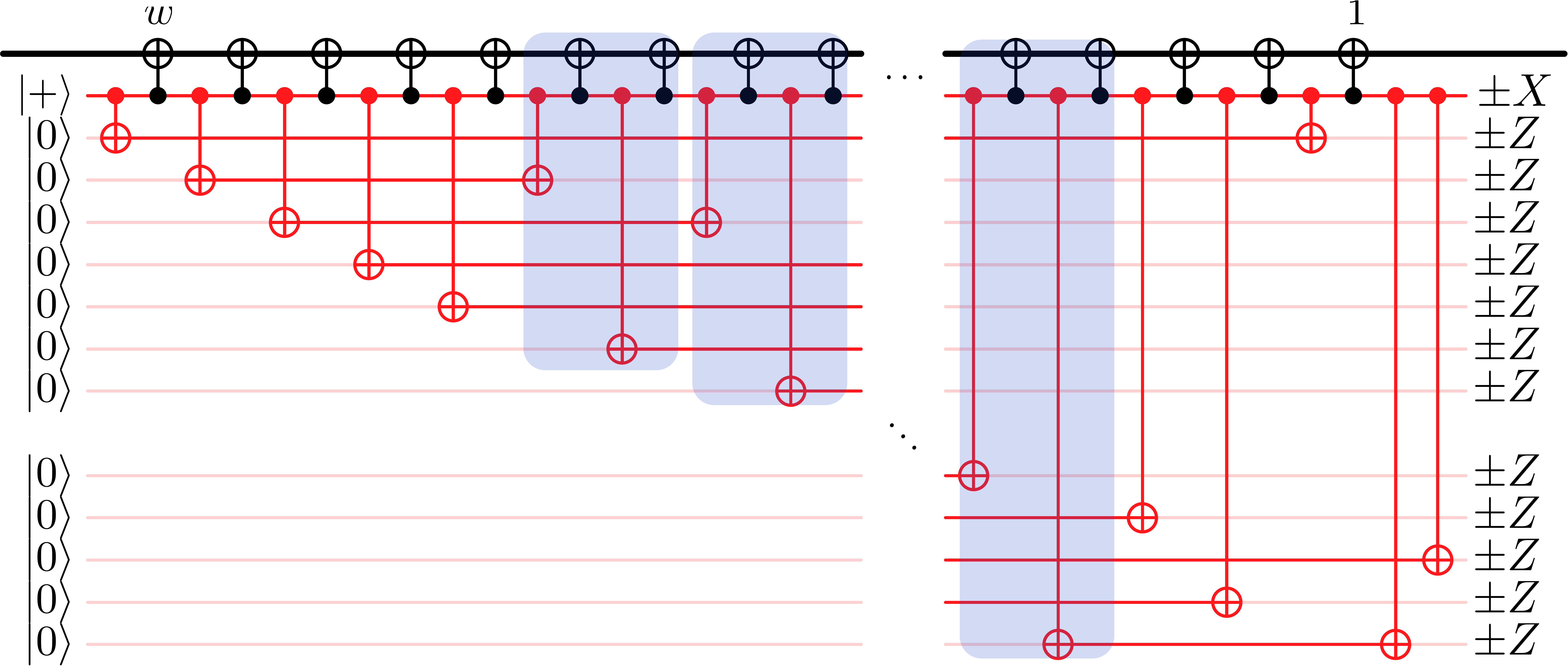} 
\caption{Distance-five syndrome measurement with slow qubit reset for a weight-$w$ $X$ stabilizer.  The thick black wire indicates a register of $w$ qubits.  An opaque red wire implies the flag is currently inactive and not catching faults. The gates in the blue section can be repeated to construct stabilizer measurement circuits for arbitrary stabilizer weight $w$. At any instant, only five flags are active. Hence this circuit can be performed with fast qubit reset using only five flag qubits. } \label{f:d5arbwsyndmeas}
\end{figure*}

\begin{figure*}
\centering
\includegraphics[width=.96\textwidth]{dist7w17} 
\caption{Distance-seven syndrome measurement with slow qubit reset for a weight-$17$ $X$ stabilizer. At any instant, only seven flags are active. Hence this circuit can be performed with fast qubit reset using only seven flag qubits. } \label{f:d7w17}
\end{figure*}

\end{document}
